\documentclass[12pt]{article}
\usepackage{amsmath}
\usepackage[round]{natbib}
\usepackage{amssymb,color}
\usepackage{appendix}
\usepackage{amsthm}
\usepackage[hyperfootnotes=false]{hyperref}
\usepackage{enumitem}
\usepackage{microtype}
\usepackage{dsfont}

\usepackage{mathtools} 
\mathtoolsset{showonlyrefs=true}  

\usepackage{caption}

\usepackage[margin=0.7in]{geometry}

\numberwithin{equation}{section}
\usepackage{comment}

\usepackage{titling}
\newtheorem{thm}{Theorem}[section]
\newtheorem{cor}[thm]{Corollary}
\newtheorem{defi}{Definition}[section]
\newtheorem{assume}{Assumption}[section]
\newtheorem{prop}[thm]{Proposition}

\newtheorem{lemma}[thm]{Lemma}

\newtheorem{remark}{Remark}[section]

\begin{document}

\author{%
Jaehyun Kim\textsuperscript{1}\thanks{jaehyun107@snu.ac.kr}
\quad and\quad Hyungbin Park\textsuperscript{2}\thanks{hyungbin@snu.ac.kr, hyungbin2015@gmail.com}\\[1em]
\normalsize{\textsuperscript{1}Department of Statistics and Data Science,}\\
\normalsize{The Chinese University of Hong Kong, Shatin, N.T., Hong Kong}\\[0.5em]
\normalsize{\textsuperscript{2}Department of Mathematical Sciences and Research Institute of Mathematics,}\\
\normalsize{Seoul National University, Seoul, South Korea}
}
\title{Quadratic $G$-BSDEs for Bond Pricing with Endogenous Short-Rate Feedback}

\date{September 16, 2026}
\maketitle

\begin{abstract}
We study robust bond valuation with endogenous short-rate feedback under volatility uncertainty. Within the $G$-expectation framework, the dependence of the short rate on the bond price yields a nonlinear fixed-point problem, represented by a quadratic $G$-BSDE for the logarithmic price. Under suitable assumptions, we establish existence, uniqueness, comparison, and stability for bounded finite-horizon solutions. An additional strict monotonicity condition yields a unique bounded infinite-horizon solution and exponential convergence of finite-horizon approximations on compact time intervals. We apply these results to inverse short-rate design, constructing discount-rate coefficients that reproduce admissible smooth bond-price targets at a fixed maturity. For long maturities, we construct feedback rules under which the compensated logarithmic price converges exponentially to a prescribed bounded state-dependent profile, while the asymptotic yield equals a specified target.
\end{abstract}

\section{Introduction}
\label{sec:intro}

The standard short-rate approach values a zero-coupon bond by discounting its terminal payoff using a short-rate process specified independently of the price being computed. A different problem arises when interest-rate decisions respond to bond-market valuations. Bond prices then depend on the future path of interest rates, while the interest-rate rule itself depends on bond prices. This circular dependence turns valuation into a self-consistency problem. When the law of the underlying economic state is also uncertain, the problem combines two distinct nonlinearities: endogenous price--rate feedback and valuation under multiple volatility scenarios. This paper develops a $G$-BSDE framework for analyzing their interaction.

Endogenous feedback between financial valuations and economic decisions is a recurring theme in economic theory. \citet{kiyotaki1997credit} show how the interaction between asset prices and collateral constraints can amplify and propagate shocks. In a continuous-time macroeconomic model with a financial sector, \citet{brunnermeier2014macroeconomic} analyze nonlinear feedback between financial conditions, asset valuations, and endogenous risk. Focusing more directly on interest-rate decisions, \citet{bernanke1997inflation} examine the interaction between monetary policy and private-sector forecasts, emphasizing the difficulties that arise when the information used to guide policy responds to the policy rule itself. These contributions motivate a joint analysis of financial variables and the rules that respond to them. Our model isolates a reduced-form price--rate feedback mechanism; it does not derive that mechanism from a general equilibrium or a welfare-maximization problem.

The treatment of uncertainty is likewise important. Specifying stochastic volatility under a single probability measure allows volatility to fluctuate, but still prescribes its probabilistic law. The $G$-expectation framework instead accommodates a family of volatility scenarios through a sublinear expectation, with $G$-Brownian motion as the driving noise \citep{denis2011function,peng2019nonlinear}. The resulting quadratic variation is itself uncertain. This feature permits the discounting rule to distinguish between ordinary time and realized quadratic variation, and it leads naturally to nonlinear valuation equations.

The financial implications of volatility ambiguity have been studied from both hedging and equilibrium perspectives. \citet{epstein2013ambiguous} develop a continuous-time asset-pricing framework with ambiguity about volatility and examine its implications for hedging and equilibrium asset returns. \citet{vorbrink2014financial} formulates financial markets with volatility uncertainty using $G$-Brownian motion and studies upper and lower arbitrage-free prices. These works clarify why robust valuation need not select a unique price independently of the pricing criterion, preferences, or market structure. In the present paper, uniqueness refers to a self-consistent valuation under a specified conditional sublinear expectation, rather than uniqueness across all possible ambiguity-sensitive pricing rules.

For bond markets, \citet{holzermann2021hull} constructs a Hull--White model under volatility uncertainty. \citet{holzermann2022term} subsequently develops a Heath--Jarrow--Morton framework driven by $G$-Brownian motion. Of particular relevance is Theorem~3.1 of that paper: under suitable regularity assumptions, a drift condition involving a market price of risk and additional market prices of uncertainty yields a sublinear expectation under which discounted bonds are symmetric $G$-martingales, and hence excludes arbitrage in the specified market. The condition also determines the risk-neutral forward-rate drift from the volatility structure. The appearance of several market prices reflects the distinction between ordinary drift and quadratic-variation-dependent drift under volatility uncertainty. This is a sufficient no-arbitrage result; its symmetry condition is stronger than the $G$-martingale property alone.

Further developments include \citet{holzermann2024pricing}, who studies the pricing of interest-rate derivatives under volatility uncertainty, and \citet{akhtari2024cox}, who establish a $G$-expectation formulation of the Cox--Ingersoll--Ross process and investigate its mathematical properties. Together, these contributions provide models and analytical tools for interest rates under uncertain volatility. Our focus is different: rather than taking the discount-rate coefficients as independent inputs to valuation, we allow them to depend on the logarithmic value being determined. The key questions are whether the resulting implicit pricing equation is well posed and how the feedback rule shapes its solution.

To describe the problem, let $X$ be a state process driven by a $d$-dimensional $G$-Brownian motion $B$, and let $\hat{\mathbb E}_s$ denote the conditional sublinear expectation used for valuation. Fix a reference maturity $T$ and write $Y_s^T=\log P(s,T)$. We consider generalized discount-rate coefficients $r(s,X_s,Y_s^T)$ and $\rho_{ij}(s,X_s,Y_s^T)$, acting through ordinary time and quadratic covariation, respectively. The self-consistency requirement is
\[
\begin{aligned}
e^{Y_s^T}=\hat{\mathbb E}_s\left[
 e^{-\int_s^T r(u,X_u,Y_u^T)\,du
 -\sum_{i,j=1}^d\int_s^T
 \rho_{ij}(u,X_u,Y_u^T)\,d\langle B^i,B^j\rangle_u}
\right],
\qquad 0\leq s\leq T.
\end{aligned}
\]
Setting $\rho_{ij}=0$ recovers feedback through the ordinary short rate alone. In the general formulation, the $\rho_{ij}$ are coefficients of the discounting rule; they are not the market prices of uncertainty appearing in the drift condition of \citet{holzermann2022term}.

The logarithmic transformation turns this fixed-point problem into the quadratic $G$-BSDE
\[
\begin{aligned}
Y_s^T={}&-\int_s^T r(u,X_u,Y_u^T)\,du\\
&+\sum_{i,j=1}^d\int_s^T
\left(\frac12 Z_u^{T,i}Z_u^{T,j}-\rho_{ij}(u,X_u,Y_u^T)\right)
\,d\langle B^i,B^j\rangle_u\\
&-\int_s^T Z_u^T\,dB_u-(K_T^T-K_s^T),
\qquad Y_T^T=0,
\end{aligned}
\]
where $K^T$ is a decreasing $G$-martingale. The quadratic term in $Z^T$ comes from the logarithmic transformation, whereas $K^T$ records the nonlinear martingale component associated with volatility uncertainty. This representation makes it possible to study self-consistent valuation using estimates and comparison arguments for quadratic backward equations. The valuation problem is formulated for a fixed reference maturity; compatibility across simultaneously traded maturities under one common money-market account is a separate requirement.

The relevant mathematical foundations are provided by the literature on $G$-BSDEs. \citet{hu2014backward,hu2014comparison} establish the Lipschitz theory and develop comparison, nonlinear Feynman--Kac, and Girsanov results. \citet{hu2018quadratic} study quadratic $G$-BSDEs through discrete approximations and fully nonlinear partial differential equations. For the quadratic and long-horizon questions considered here, several extensions are particularly relevant. \citet{hu2018ergodic} develop infinite-horizon and ergodic $G$-BSDEs and their applications. \citet{hu2022quadratic} treat quadratic $G$-BSDEs with convex generators and unbounded terminal conditions. \citet{sun2024forward} connect ergodic quadratic $G$-BSDEs with homothetic forward performance processes under volatility uncertainty. More recently, \citet{lin2026infinite} establish well-posedness and comparison results for infinite-horizon quadratic $G$-BSDEs.

The work of Falei Wang and his coauthors provides a central methodological
foundation for our analysis. In particular, our treatment of infinite-horizon
equations builds on the linearization and finite-horizon approximation methods
of \citet{hu2018ergodic}. Their linearization result is used explicitly in the
proof of Lemma~\ref{lem:infinite-linearization-bmo}, where it is combined with
truncation and $G$-BMO estimates to handle quadratic growth. These methods
underpin the well-posedness and horizon-convergence analysis in
Section~\ref{section 4}, which we apply to self-consistent bond valuation and
long-maturity short-rate design.

Building on this theory, the present paper addresses two connected financial questions: the well-posedness of endogenous bond valuation and the inverse design of short-rate feedback for prescribed price targets. We first develop a Markovian analysis suited to the pricing application. With bounded terminal data, appropriate Lipschitz assumptions, and quadratic growth in the martingale integrand, we establish a priori bounds and BMO estimates, together with a comparison principle, stability, and the existence and uniqueness of bounded solutions. We also establish spatial Lipschitz continuity and temporal $1/2$-H\"older continuity of the associated Markovian value function. Applied to the logarithmic pricing equation, these results yield a unique bounded self-consistent valuation under the stated feedback assumptions.

We next consider infinite-horizon equations under a strict monotonicity condition in the backward variable. The corresponding estimates yield a unique bounded infinite-horizon solution and show that finite-horizon approximations converge exponentially on every compact time interval. Thus the influence of a remote terminal condition can be controlled quantitatively. This horizon-stability result is the analytical link between finite-maturity valuation and the long-maturity design problem.

Finally, we study the inverse question of constructing discount-rate coefficients from a prescribed price target. For a fixed maturity, an admissible smooth function $\psi$, satisfying $\psi(T,x)=0$, determines coefficients that reproduce $P(s,T)=e^{\psi(s,X_s)}$. For long maturities, we prescribe a target yield $\lambda>0$ and an admissible smooth bounded profile $\phi$, and use feedback through the compensated logarithmic price
\[
\widetilde Y_s^T:=\log P(s,T)+\lambda(T-s).
\]
A linear restoring term $\mu(y-\phi(x))$, with $\mu>0$, together with suitable state-dependent discount-rate coefficients, identifies $\phi(X_s)$ as the bounded infinite-horizon profile. Under the required regularity and boundedness conditions, the finite-horizon solutions satisfy
\[
\sup_{0\leq u\leq s}
\left|\widetilde Y_u^T-\phi(X_u)\right|
\leq C e^{-\mu(T-s)},\qquad T\geq s,
\quad\text{quasi-surely}.
\]
Consequently, for each fixed $s$,
\[
\begin{gathered}
e^{\lambda(T-s)}P(s,T)\longrightarrow e^{\phi(X_s)},\\
-\frac{\log P(s,T)}{T-s}\longrightarrow\lambda
\qquad\text{as }T\to\infty.
\end{gathered}
\]
Here $\lambda$ is prescribed as a design target, not determined as an ergodic eigenvalue. Likewise, $\mu$ controls convergence as maturity increases; the result is not a statement about the speed of adjustment following a policy intervention in calendar time.

The remainder of the paper is organized as follows. Section~2 recalls the $G$-expectation framework and the BMO tools used in the analysis. Section~3 studies finite-horizon quadratic $G$-BSDEs, and Section~4 develops the infinite-horizon theory and convergence estimates. Section~5 applies these results to bond valuation with endogenous short-rate feedback. Section~6 considers fixed-maturity and long-maturity target design. Section~7 concludes.

\section{Preliminaries}\label{section 2}
\subsection{\textit{G}-expectation}

In this subsection, we briefly recall the basic notions and results of
$G$-expectation theory. We refer the reader to
\cite{denis2011function}, \cite{hu2014backward}, and
\cite{peng2019nonlinear} for further details.
Let $\Omega=C_0([0,\infty),\mathbb{R}^d)$ be the canonical space of all
$\mathbb{R}^d$-valued continuous paths
$(\omega_t)_{t\in[0,\infty)}$ starting from $\omega_0=0$, equipped with
the metric
$$
\rho(\omega^{(1)},\omega^{(2)})
:=
\sum_{n=1}^\infty
\frac{1}{2^n}
\left(
1\wedge
\max_{t\in[0,n]}
|\omega_t^{(1)}-\omega_t^{(2)}|
\right).
$$
Let $B$ denote the canonical process on $\Omega$, that is,
$B_t(\omega)=\omega(t)$ for $t\geq0$.
For each $t>0$, define
\begin{equation}
\begin{split}
\textnormal{Lip}(\Omega_t)
:=\Big\{&\varphi\big(
B_{t_1}-B_{t_0},B_{t_2}-B_{t_1},\cdots,B_{t_n}-B_{t_{n-1}}
\big)\,:\ n\in\mathbb{N},\\
&0\leq t_0\leq\cdots\leq t_n\leq t,\quad
\varphi\in C_{l,\textnormal{Lip}}(\mathbb{R}^{d\times n})\Big\}.
\end{split}
\end{equation}
Here, $C_{l,\textnormal{Lip}}(\mathbb{R}^{d\times n})$ denotes the
space of all continuous functions $\varphi$ for which there exist
constants $C>0$ and $k\in\mathbb{N}$, depending only on $\varphi$, such
that
$|\varphi(x)-\varphi(y)|
\leq C(1+|x|^k+|y|^k)|x-y|$
for all $x,y\in\mathbb{R}^{d\times n}$.
It is immediate that
$\textnormal{Lip}(\Omega_t)\subseteq\textnormal{Lip}(\Omega_T)$
whenever $t\leq T$. We therefore set
$\textnormal{Lip}(\Omega):=\bigcup_{t\geq0}\textnormal{Lip}(\Omega_t)$.

Let $\mathbb{S}^d$ denote the space of all $d\times d$ symmetric
matrices.
Given a monotone sublinear functional
$G:\mathbb{S}^d\to\mathbb{R}$, there exists a sublinear expectation
space $(\Omega,\textnormal{Lip}(\Omega),\hat{\mathbb{E}})$ under which
the canonical process $B$ is a $G$-Brownian motion satisfying
\begin{equation}
G(A)
=
\frac{1}{2}
\hat{\mathbb{E}}
\left[
\langle AB_1,B_1\rangle
\right],
\qquad A\in\mathbb{S}^d.
\end{equation}
Moreover, there exists a bounded and closed set $\Gamma$ of
$d\times d$ symmetric positive semidefinite matrices such that
\begin{align}
G(A)
=
\frac{1}{2}
\sup_{Q\in\Gamma}
\operatorname{tr}(AQ).
\label{eqn: gamma}
\end{align}
The conditional $G$-expectation $\hat{\mathbb{E}}_t$ is initially
defined for $X\in\textnormal{Lip}(\Omega)$ by
$$
\hat{\mathbb{E}}_t[X]
=
\left.
\hat{\mathbb{E}}
\left[
\varphi
\left(
x_1,\cdots,x_j,
B_{t_{j+1}}-B_{t_j},
\cdots,
B_{t_n}-B_{t_{n-1}}
\right)
\right]
\right|_{
\substack{
x_1=B_{t_1}-B_{t_0},\ \cdots,\\
x_j=B_{t_j}-B_{t_{j-1}}
}
}.
$$
Here, $X$ is represented as
$X=\varphi(B_{t_1}-B_{t_0},B_{t_2}-B_{t_1},\cdots,
B_{t_{j+1}}-B_{t_j},\cdots,B_{t_n}-B_{t_{n-1}})$,
where
$\varphi\in C_{l,\textnormal{Lip}}(\mathbb{R}^{d\times n})$
and $t=t_j$ for some $j=0,\cdots,n-1$.
The sublinear expectation space
$(\Omega,\textnormal{Lip}(\Omega),\hat{\mathbb{E}})$ is called a
$G$-expectation space, while $\hat{\mathbb{E}}$ and
$\hat{\mathbb{E}}_t$ are referred to as the $G$-expectation and the
conditional $G$-expectation, respectively.


For $p\geq 1$, we define $\mathbb{L}_G^p(\Omega_t)$ ($\mathbb{L}_G^p(\Omega)$, respectively) as the completion of $Lip(\Omega_t)$ ($Lip(\Omega)$, respectively) with respect to the norm $|\!|\xi|\!|_{p,G}:= (\hat{\mathbb{E}}[|\xi|^p])^{\frac{1}{p}}$. For each $t\geq 0$, the conditional sublinear expectation extends continuously to $\hat{\mathbb{E}_t}:\mathbb{L}^p_G(\Omega)\to \mathbb{L}^p_G(\Omega_t)$.
In this paper, we shall only
consider   $G$-Brownian motions satisfying the strong ellipticity condition, that is, 
there exist  strictly positive constants $\overline{\sigma}, \underline{\sigma}$ such that
\begin{equation}
\frac{1}{2}\underline{\sigma}^2\text{tr}(A-B)\leq G(A)-G(B)\leq \frac{1}{2}\overline{\sigma}^2\text{tr}(A-B)\text{ for } A\geq B\,.
\end{equation}

The conditional $G$-expectation admits a probabilistic representation.
For the following representation theorem, see
\citet{hu2009representation,hu2021extended}.
 Let $(\mathcal{F}_t)_{t\ge0}$ be the filtration generated by the canonical process $B$ and define $\mathcal{F}=\sigma(\cup_{t\ge0}\mathcal{F}_t)$.

\begin{thm}
The $G$-expectation $\hat{\mathbb{E}}$ can be represented by a
weakly compact family $\mathcal{P}$ of probability measures on
$(\Omega,\mathcal{F})$. More precisely,
\begin{equation}
\hat{\mathbb{E}}[X]
=
\sup_{P\in\mathcal{P}}E^P[X]
\text{ for all }X\in\mathbb{L}^1_{\textit{G}}(\Omega).
\end{equation}
The family $\mathcal{P}$ is called a representing set for
$\hat{\mathbb{E}}$. Furthermore, for every $P\in\mathcal{P}$,
the conditional $G$-expectation has the following representation:
\begin{align}
\hat{\mathbb{E}}_t[X]
=
\operatorname*{ess\,\sup}_{Q\in\mathcal{P}(t,P)}
E^Q[X\mid\mathcal{F}_t]
\hspace{1cm}P\textnormal{-almost surely},
\end{align}
where
\[
\mathcal{P}(t,P)
:=
\left\{
Q\in\mathcal{P}
\,\middle|\,
E^Q[X]=E^P[X]
\text{ for all }X\in Lip(\Omega_t)
\right\}.
\]
\end{thm}

We define the capacity as 
\begin{equation}
c(A):= \sup_{P\in \mathcal{P}} P(A),\hspace{0.2cm} A\in \mathbb{F}
\end{equation}
where $ \mathbb{F}:= \bigcup\limits_{t\geq 0} \mathcal{F}_t$. 
A set $A\in \mathbb{F}$ is called polar if $c(A)=0$. We say that a property holds ``quasi-surely'' if it holds outside a polar set. 
Let $\mathcal{B}_b(\Omega)$ be the set of all bounded $\mathcal{F}$-measurable real-valued functions, and define $\mathbb{L}^p_b(\Omega)$ as the completion of $\mathcal{B}_b(\Omega)$ with respect to the norm $|\!| \cdot|\!|_{p,G}$. We define $\mathbb{L}_G^{\infty}(\Omega_t)$ ($\mathbb{L}_G^{\infty}(\Omega)$, respectively) as the completion of $Lip(\Omega_t)$ ($Lip(\Omega)$, respectively) with respect to the norm $|\!|\xi|\!|_{\infty,G}:=\inf\{C\ge 0\,|\, |\xi|\le C \textnormal{ quasi-surely } \}$. \citet{denis2011function} proved that  
\begin{align}
    \mathbb{L}^p_b(\Omega)=\{X|X\text{ is }\mathcal{F}\text{ measurable, } \hat{\mathbb{E}}[|X|^p]<\infty\text{ and } \lim_{n\rightarrow\infty}\hat{\mathbb{E}}[|X|^p\mathds{1}_{|X|>n}]=0\}.
\end{align}
Consequently, to establish $Y\in \mathbb{L}_b^1(\Omega)$, it suffices to show that $Y\in \mathbb{L}_G^p(\Omega)$ for some $p>1$.

We now consider several process spaces and norms. 
For $T>0$ and $p\ge 1$, define
\begin{align}
\mathbb{M}^{p,0}(0,T)&=\big\{ \eta_t(\omega)=\sum_{k=1}^{n}\xi_k(\omega) \mathds{1}_{[t_k,t_{k+1})}(t) \,|\,  n \in \mathbb{N},\xi_k\in \mathbb{L}_G^p(\Omega_{t_k}), 0\le t_k< t_{k+1}\leq T, k=1,2,\cdots,n\big\}\\
\mathbb{S}^0(0,T)&=\big\{(\varphi(\,\cdot\,,B_{t_1\wedge \cdot},\cdots,B_{t_n\wedge \cdot})\,|\, n\in\mathbb{N},   t_1,\cdots,t_n\in[0,T], \varphi\in C_{b,Lip}(\mathbb{R}^{n+1}) \big\}
\end{align}
with norms
\begin{align}
|\!|\eta|\!|_{\mathbb{M}^p}=(\hat{\mathbb{E}}[\int_0^T|\eta_u|^p\,du])^{\frac{1}{p}}\,,\,\,
|\!|\eta|\!|_{\mathbb{H}^p}=(\hat{\mathbb{E}}[(\int_0^T|\eta_u|^2\,du)^{\frac{p}{2}}])^{\frac{1}{p}}\,,\,\,
|\!|\eta|\!|_{\mathbb{S}^p}=(\hat{\mathbb{E}}[\sup_{0\leq u\leq T}|\eta_u|^p])^{\frac{1}{p}}\,.
\end{align}
Here, $C_{b,Lip}(\mathbb{R}^{n+1})$ is the space of all bounded Lipschitz continuous functions on $\mathbb{R}^{n+1}.$ 
The completions of $\mathbb{M}^{p,0}(0,T)$ with respect to the norms
$|\!|\cdot|\!|_{\mathbb{M}^p}$ and $|\!|\cdot|\!|_{\mathbb{H}^p}$
are denoted by $\mathbb{M}^p(0,T)$ and $\mathbb{H}^p(0,T)$,
respectively. The completion of $\mathbb{S}^{0}(0,T)$ with respect
to the norm $|\!|\cdot|\!|_{\mathbb{S}^p}$ is denoted by
$\mathbb{S}^p(0,T)$. 
The corresponding spaces $\mathbb{M}^p(t,T)$, $\mathbb{H}^p(t,T)$, and
$\mathbb{S}^p(t,T)$ are defined analogously for $0\leq t<T$.
For $d\in\mathbb{N}$, the spaces $\mathbb{M}^p(t,T;\mathbb{R}^d)$,
$\mathbb{H}^p(t,T;\mathbb{R}^d)$, and $\mathbb{S}^p(t,T;\mathbb{R}^d)$
are defined as the $d$-fold product spaces
$(\mathbb{M}^p(t,T))^d$, $(\mathbb{H}^p(t,T))^d$,
and $(\mathbb{S}^p(t,T))^d$, respectively.
Moreover, for $t\geq0$, we define
$\mathbb{M}^p(t,\infty;\mathbb{R}^d):=\bigcap_{T>t}\mathbb{M}^p(t,T;\mathbb{R}^d)$,
$\mathbb{H}^p(t,\infty;\mathbb{R}^d):=\bigcap_{T>t}\mathbb{H}^p(t,T;\mathbb{R}^d)$,
and
$\mathbb{S}^p(t,\infty;\mathbb{R}^d):=\bigcap_{T>t}\mathbb{S}^p(t,T;\mathbb{R}^d)$.

Following \citet{li2011stopping}, for $p\geq 1$ and
$\eta\in \mathbb{H}^p(t,T)$, we define the integrals
$\int_t^T\eta_u\,dB_u$ and $\int_t^T\eta_u\,d \langle B\rangle_u$.
We next introduce $G$-martingales and recall some of their properties.
\begin{defi}
    A process $M=(M_s)_{s\ge0}$ is called a $G$-martingale if
$M_s\in \mathbb{L}^1_G(\Omega_s)$ for $s\ge 0$ and
$\hat{\mathbb{E}}_s[M_t]=M_s$ for $0\le s\le t$.
A process $M$ is called a symmetric $G$-martingale if both $M$
and $-M$ are $G$-martingales.
\end{defi}

\subsection{BMO martingales}
 This subsection introduces $G$-BMO martingale generators and recalls
some of their properties. We refer the reader to
\citet{hu2018quadratic} for further details.
\begin{defi}
	Let $0\le t<T$ and $Z\in \mathbb{H}^2(t,T;\mathbb{R}^d).$
	We say  $Z$ is a $G$-BMO martingale generator if $Z$ satisfies  
	\begin{align}
	|\!|Z|\!|^2_{\textnormal{BMO}_G}:=\sup_{P\in \mathcal{P}}\sup_{\tau\in 
		\mathcal{T}_{[t,T]}}\Big|\!\Big| \mathbb{E}^P_{\tau}\Big[\int_{\tau}^TZ^i_uZ_u^j\,d\langle B^i,B^j\rangle_u\Big]\Big|\!\Big|_{L^\infty(P)} <\infty
	\end{align}    
	where $\mathcal{T}_{[t,T]}$ denotes the set of all $\mathcal{F}$-stopping times taking values in $[t,T].$
\end{defi}

We can define a new sublinear expectation and sublinear conditional expectation on a $G$-expectation space $(\Omega, Lip(\Omega), \hat{\mathbb{E}})$ using a $G$-BMO martingale generator.
For fixed $T>0,$ let  $Z\in \mathbb{H}^2(0,T;\mathbb{R}^d)$ be a $G$-BMO martingale generator. Then the process 
 \begin{align}
 M_t^Z:=e^{\int_0^tZ^i_u\,dB^i_u-\frac{1}{2}\int_0^t Z_u^ iZ_u^j\,d\langle B^i,B^j \rangle _u}\,,\; 0\le t\le T
 \end{align}
is a symmetric $G$-martingale.  
 Define a sublinear expectation $\hat{\mathbb{E}}^Z$ and a sublinear conditional expectation $\hat{\mathbb{E}}^Z_t$ as
\begin{align}\label{eqn:change}
\hat{\mathbb{E}}^Z[X]=\hat{\mathbb{E}}[M^Z_T X] \,,\;\hat{\mathbb{E}}^Z_t[X]=(M_t^Z)^{-1}\hat{\mathbb{E}}_t[M^Z_T X] 
\end{align} 
for $X\in \textnormal{Lip}(\Omega)$, respectively.  
We say  $\hat{\mathbb{E}}^Z$ ($\hat{\mathbb{E}}^Z_t$, respectively) is the sublinear expectation (the sublinear conditional expectation, respectively)   induced by $Z.$
The sublinear expectations
$\hat{\mathbb{E}}$ and $\hat{\mathbb{E}}^Z$ are equivalent.
Moreover, the process
\begin{equation}
\label{new_GBM}
B^Z=(B^{Z,i})_{1\le i\le d}:=\Big(B^i-\int_0^\cdot Z^j_u\,d\langle B^i,B^j\rangle _u\Big)_{1\le i\le d}
\end{equation}  is a $G$-Brownian motion under the sublinear expectation $\hat{\mathbb{E}}^Z$. 
It follows that
the quadratic variations of $B$ and $B^Z$ coincide under $\hat{\mathbb{E}}$ and $\hat{\mathbb{E}}^Z.$
For further details, see
\citet{xu2011girsanov,hu2018quadratic}.
\begin{lemma}\label{reverse holder inequality} 
	Define a function $\phi(x)=\big( 1+\frac{1}{x^2} \log{\frac{2x-1}{2(x-1)}}\big)^\frac{1}{2}-1.$ If $Z\in \mathbb{H}^2(0,T;\mathbb{R}^d)$ is a $G$-BMO martingale generator with $|\!|Z|\!|_{\textnormal{BMO}_G}\leq \phi(q)$ for $q>1$, then there is a constant $C_q>0$ such that 
	\begin{align}
	\sup_{P\in \mathcal{P}}\sup_{\tau\in \mathcal{T}_{[0,T]}} \Big|\!\Big|\mathbb{E}^{P}_{\tau}\Big[\Big(\frac{M_T^Z}{M_{\tau}^Z}\Big)^q\Big]\Big|\!\Big|_{L^\infty(P)}\leq C_q\,.
	\end{align}
	The exponent $q$ is called the order of the reverse H\"older inequality for $M^Z$. In addition, if 
	$K$ is a decreasing $G$-martingale such that $K_0=0$ and  $K_{T}\in \mathbb{L}_G^p(0,T)$ for  $p>\frac{q}{q-1}$, then $K$ is a decreasing $G$-martingale under $\hat{\mathbb{E}}^Z$.
\end{lemma}

\begin{lemma}\label{K is MG in BMO}
   Let $Z$ be a $G$-BMO martingale generator, let $q>1$ be the order
of the reverse H\"older inequality for $M^Z$, and let $K$ be a
decreasing $G$-martingale such that $K_0=0$ and
$K_t\in \mathbb{L}^p(0,t)$ for $p>\frac{q}{q-1}$.
Then $K$ is a decreasing $G$-martingale under
$\hat{\mathbb{E}}^Z$.
\end{lemma}

\section[Finite-horizon quadratic G-BSDEs]{Finite-horizon quadratic $G$-BSDEs}\label{section 3}

Consider the SDE and $G$-BSDE 
\begin{align}
\label{SDE} X_s^{t,\xi}&=\xi+\int_t^s b(u,X_u^{t,\xi})\,du+\int_t^s h_{ij}(u,X_u^{t,\xi})\,d\langle B^i,B^j\rangle_u+\int_t^s \sigma_j(u,X_u^{t,\xi})\,dB^j_u
\\
\label{QBSDE}Y_s^{t,\xi}&=\Phi(X_T^{t,\xi})+\int_s^Tf(u,X_u^{t,\xi},Y_u^{t,\xi},Z_u^{t,\xi})\,du+\int_s^Tg_{ij}(u,X_u^{t,\xi},Y_u^{t,\xi},Z_u^{t,\xi})\,d\langle B^i,B^j\rangle_u\\&-\int_s^TZ_u^{t,\xi}\,dB_u-(K_T^{t,\xi}-K_s^{t,\xi}) 
\end{align}
for $t\le s\le T$, $\xi\in \mathbb{L}_G^1(\Omega_t;\mathbb{R}^m),$ $b,h_{ij},\sigma_j:[0,T]\times\mathbb{R}^m\mapsto \mathbb{R}^m$, 
$\Phi:\mathbb{R}^m\mapsto \mathbb{R},$   $f,g_{ij}:[0,T]\times \mathbb{R}^m\times \mathbb{R}\times \mathbb{R}^d\mapsto 
\mathbb{R}.$
Occasionally, we write $X^{t,\xi},Y^{t,\xi},Z^{t,\xi},K^{t,\xi}$ as $X,Y,Z,K,$ respectively, omitting the superscripts $t,\xi.$

\begin{assume}\label{assumption} Assume the functions $b,h_{ij}, \sigma_j$ satisfy the following properties.  
	\begin{enumerate}
		\item For $1\leq i,j \leq d$, $h_{ij}=h_{ji}$.    \label{finite assumption first}
		\item The functions $b,h_{ij},\sigma_j $ are continuous in $s$.\label{finite assumption second}
		\item There is a constant $C_1>0$ such that 
		\begin{align}
		&|b(s,x)-b(s,x')|+\sum_{i,j=1}^d|h_{ij}(s,x)-h_{ij}(s,x')|+\sum_{j=1}^d|\sigma_{j}(s,x)-\sigma_{j}(s,x')|\leq C_1|x-x'|
		\end{align}
		for $s\in[t,T],$ $x,x'\in \mathbb{R}^m$.
		\label{finite assmption third}
		
	\end{enumerate}
\end{assume}

The following theorem states the existence, uniqueness and regularity of solutions to the SDE \eqref{SDE}.

\begin{thm}\label{SDE sublinear property} 
	Suppose Assumption~\ref{assumption} holds. Then for $t>0$ and $\xi\in \mathbb{L}_G^p(\Omega_t;\mathbb{R}^m)$, 
there exists a unique solution $X^{t,\xi}$ to \eqref{SDE}. 
	For any $G$-BMO martingale generator $Z$ and $p\geq 2$,
    there exists a constant $C$ depending only on $C_1$, $\overline{\sigma}$, $p$, $T$, and the order $q$ of the reverse H\"older inequality for $M^Z$ 
      such that 
	 	\begin{align}
	 &\hat{\mathbb{E}}^Z_t[|X_{t+\delta}^{t,\xi}-X_{t+\delta}^{t,\xi'}|^p]\leq C|\xi-\xi'|^p\,,\\
	 &\hat{\mathbb{E}}^Z_t[|X_{t+\delta}^{t,\xi}|^p]\leq C(1+|\xi|^p)\,,\\
	 &\hat{\mathbb{E}}^Z_t[\sup_{s\in [t,t+\delta]}|X_s^{t,\xi}-\xi|^p]\leq C(1+|\xi|^p)\delta^\frac{p}{2}\,
	 \end{align}
	 for all $\xi,\xi'\in \mathbb{L}_G^p(\Omega_t;\mathbb{R}^m)$ 
  and $\delta\in[0,T-t].$ 
\end{thm}

\begin{proof}
Existence and uniqueness follow from \cite[Theorem 5.1.3]{peng2019nonlinear}. The three inequalities are consequences of Lemma~\ref{reverse holder inequality} and the results in Chapter~5 of \cite{peng2019nonlinear}.
\end{proof}

We now focus on the BSDE \eqref{QBSDE}. 
We first state the assumptions on the drivers $f$ and $g_{ij}$ and the terminal condition $\Phi$.

\begin{assume}\label{assumption2} Assume the functions $\Phi, f,g_{ij}$ satisfy the following properties.  
	\begin{enumerate}
		\item For $1\leq i,j \leq d$, $g_{ij}=g_{ji}.$    
		\item The functions $f$ and $g$ are continuous in $s$.
		\item There is a constant $C_1>0$ such that 
		\begin{align}
    	&|\Phi(x)-\Phi(x')|\leq C_1|x-x'|\,,\\
		&|f(s,x,y,z)-f(s,x',y',z')|+\sum_{i,j=1}^d|g_{ij}(s,x,y,z)-g_{ij}(s,x',y',z')|\\ &\hspace{2cm}\leq C_1(|x-x'|+|y-y'|+(1+|z|+|z'|)|z-z'|)
		\end{align}
		for $s\in[t,T],$ $x,x'\in \mathbb{R}^m$, $y,y'\in \mathbb{R}$,   $z,z'\in \mathbb{R}^d.$
		
		\item There is a constant $C_2>0$ such that
		\begin{align}
		|f(s,x,0,0)|+\sum_{i,j=1}^d|g_{ij}(s,x,0,0)|+|\Phi(x)|\leq C_2    
		\end{align}
		for $s\in[t,T],$ $x\in \mathbb{R}^m$.
	\end{enumerate}
\end{assume}
\noindent Under Assumption~\ref{assumption2}, it follows that
	\begin{align}
	|f(s,x,y,z)|+\sum_{i,j=1}^d|g_{ij}(s,x,y,z)|\leq C_2+2C_1(1+|y|+|z|^2) 
	\end{align}
		for $s\in[0,T],$ $x\in \mathbb{R}^m$, $y\in \mathbb{R}$,   $z\in \mathbb{R}^d.$

\begin{defi}
A triple $(Y,Z,K)$ is called a solution of \eqref{QBSDE} on $[t,T]$ if
$(Y,Z)\in\mathbb{S}^{2}(t,T)\times
\mathbb{H}^{2}(t,T;\mathbb{R}^{d})$,
$K$ is a decreasing $G$-martingale satisfying
$K_t=0$ and $K_T\in\mathbb{L}_G^{2}(\Omega_T)$,
and \eqref{QBSDE} holds for all $s\in[t,T]$, quasi-surely.
\end{defi}

The remainder of this section studies the main properties and
existence of solutions to the BSDE. We first establish uniqueness,
a comparison principle, regularity, and stability in
Subsection~\ref{sec:prior}. We then prove the existence of solutions
in Subsection~\ref{sec:exist}. The financial application to bond
pricing with endogenous short-rate feedback is developed separately
in Section~\ref{sec:app_1}.

\subsection{A priori estimates}
\label{sec:prior}
\begin{thm}\label{Z is BMO and K is intgrable}
	Suppose Assumptions~\ref{assumption} and~\ref{assumption2} hold,
and let the triple $(Y,Z,K)=(Y^{t,\xi},Z^{t,\xi},K^{t,\xi})$
be a solution to \eqref{QBSDE}. If $|Y_s|$ is bounded by a constant for $t\leq s \leq T$, then $\big|\!\big|\sup_{t\leq u \leq T}|Y_u|\big|\!\big|_{\mathbb{L}^{\infty}}+|\!|Z|\!|_{\textnormal{BMO}_G} \leq C$ where $C$ is a constant depending only on $C_1,$ $C_2$,  $\overline{\sigma}$ and $T$. Moreover, $K\in \mathbb{L}_G^p(\Omega_T)$ for all $p\geq1$.
\end{thm}
\begin{proof}
	For simplicity, assume $d=1,m=1,$ $t=0$ and $f=0.$
	We first prove that $Z$ is a $G$-BMO martingale generator and that the norm $|\!|Z|\!|_{\textnormal{BMO}_G}$ is bounded by a constant depending only on $C_1,$ $C_2$, $\overline{\sigma}$, $T$ and  $\big|\!\big|\sup_{0\leq u \leq T}|Y_u|\big|\!\big|_{\mathbb{L}^{\infty}}$, and then show that $K\in \mathbb{L}^p$ for every $p\geq 1.$  Finally, we prove that the norm $\big|\!\big|\sup_{0\leq u \leq T}|Y_u|\big|\!\big|_{\mathbb{L}^{\infty}}$ is bounded by a constant depending only on $C_1$, $C_2$, $\overline{\sigma}$ and $T.$

    Applying It\^{o}'s formula to $e^{-\gamma Y}$ 	 
    for a constant $\gamma>0$ which will be specified later, we have 
	\begin{align}
	\frac{\gamma^2}{2}\int_{\tau}^T   e^{-\gamma Y_u}|Z_u|^2\,d\langle B \rangle _u= &\; e^{-\gamma \Phi(X_T)}-e^{-\gamma Y_{\tau}}-\gamma \int_{\tau}^T e^{-\gamma Y_u}g(u,X_u,Y_u,Z_u)\,d\langle B \rangle _u \\
	&+\int_{\tau}^T\gamma e^{-\gamma Y_u}Z_u\,dB_u+\int_{\tau}^T \gamma e^{-\gamma Y_u}\,dK_u\quad P\textnormal{-almost surely}
	\end{align}
    where $P\in \mathcal{P}$ and $\tau\in\mathcal{T}_{[0,T]}.$  
	Since $|g(s,x,y,z)|\leq C_2+2C_1(1+|y|+|z|^2),$  
	\begin{align}
	\frac{\gamma^2}{2}\int_{\tau}^T   e^{-\gamma Y_u}|Z_u|^2\,d\langle B \rangle _u&\leq e^{-\gamma \Phi(X_T)}+\gamma\int_{\tau}^T  e^{-\gamma Y_u}(C_2+2C_1(1+|Y_u|+|Z_u|^2) \,d\langle B \rangle _u \\
	&\quad+\gamma\int_{\tau}^T  e^{-\gamma Y_u}Z_u\,dB_u \quad P\textnormal{-almost surely}.
	\end{align}
	Setting $\gamma=6C_1$, 
	\begin{align}
	6C_1^2 \int_{\tau}^T  e^{-6C_1 Y_u}|Z_u|^2\,d\langle B \rangle _u&\leq e^{-6C_1 \Phi(X_T)}+6C_1\int_{\tau}^T  e^{-6C_1 Y_u}(C_2+2C_1(1+|Y_u|)\, d\langle B \rangle _u \\
	&+6C_1\int_{\tau}^T e^{-6C_1 Y_u}Z_u\,dB_u\quad P\textnormal{-almost surely}.
	\end{align}
	This implies that $Z$ is a $G$-BMO martingale generator and that $|\!|Z|\!|_{\textnormal{BMO}_G}$ is bounded by a constant depending on $C_1, C_2$, $\overline{\sigma}$, $T$ and the norm $\big|\!\big|\sup_{0\leq u \leq T}|Y_u|\big|\!\big|_{\mathbb{L}^{\infty}}.$
	 
	In order to show   $K\in \mathbb{L}^p$ for every $p\geq 1$, we first observe that  $\hat{\mathbb{E}}[(\int_0^T|Z_u|^2 d\langle B \rangle _u)^n]<\infty$ for all $n \in \mathbb{N}$. Applying It\^{o}'s formula, we have
	\begin{align}
	\Big(\int_0^T|Z_u|^2\,d\langle B \rangle _u\Big)^n&=n!\int_0^T\int_0^{t_1}\cdots\int_0^{t_{n-1}} |Z_{t_1}|^2\cdots|Z_{t_n}|^2\,d\langle B \rangle_{t_n}\cdots d\langle B \rangle_{t_1}\\
	&=n!\int_0^T\int_{t_n}^T\cdots\int_ {t_2}^T |Z_{t_1}|^2\cdots|Z_{t_n}|^2\, d\langle B \rangle_{t_1}\cdots d\langle B \rangle_{t_n}\,,
	\end{align}
	which implies that $\hat{\mathbb{E}}[(\int_0^T|Z_u|^2 d\langle B \rangle _u)^n]\leq n!(|\!|Z|\!|_{\textnormal{BMO}_G})^n<\infty$.
	Since $(Y,Z,K)$ is a solution  to \eqref{QBSDE}, there exists a constant $C>0$ such that
	\begin{align}
	|K_s|^p\leq C(|Y_s|^p+|Y_0|^p+\Big(\int_0^s(C_2+2C_1(1+|Y_u|+|Z_u|^2)\,d\langle B \rangle _u\Big)^p +\Big(\sup_{0\leq v \leq s}\Big|\int_0^vZ_udB_u\Big|\Big)^p \,.
	\end{align}
	Using \cite[Theorem 2.1]{gao2009pathwise} and $\hat{\mathbb{E}}[(\int_0^T|Z_u|^2\,d\langle B \rangle _u )^p]<\infty$, we obtain that  $\hat{\mathbb{E}}[|K_s|^p]$ is bounded by a constant  depending on  
	$C_1, C_2,\overline{\sigma}$ and the norms $\big|\!\big|\sup_{0\leq u \leq T}|Y_u|\big|\!\big|_{\mathbb{L}^{\infty}}$, $|\!|Z|\!|_{\textnormal{BMO}_G}$. 
	
	Finally, we show that the norm $\big|\!\big|\sup_{0\leq s \leq T}|Y_s|\big|\!\big|_{\mathbb{L}^{\infty}}$ is bounded by a constant depending on $C_1,C_2,\overline{\sigma}$ and $T.$ We use the linearization argument from the proof of
\cite[Proposition~3.5]{hu2018quadratic}. For   $\epsilon>0$ let $l^\epsilon$ be a Lipschitz function  satisfying $\mathds{1}_{[-\epsilon,\epsilon]}(x)\leq l^\epsilon(x)\leq \mathds{1}_{[-2\epsilon,2\epsilon]}(x)$. Define three processes 
	\begin{align}
	a_s^{\epsilon}&=(1-l^\epsilon(Y_s))\frac{g(s,X_s,Y_s,Z_s)-g(s,X_s,0,Z_s)}{Y_s}\mathds{1}_{|Y_s|\geq 0}\,,\\
	b_s^{\epsilon}&=(1-l^\epsilon(Z_s))\frac{g(s,X_s,0,Z_s)-g(s,X_s,0,0)}{|Z_s|^2}Z_s\mathds{1}_{|Z_s|\geq 0}\,,\\
	m_s^{\epsilon}&=l^\epsilon(Y_s)(g(s,X_s,Y_s,Z_s)-g(s,X_s,0,Z_s))+l^\epsilon(Z_s)(g(s,X_s,0,Z_s)-g(s,X_s,0,0))\,.
	\end{align}
	Then $|a_s^{\epsilon}|\leq C_1$, $|b_s^{\epsilon}|\leq C_1(1+|Z_s|)$ and $m_s^{\epsilon}\leq 4\epsilon C_1(1+\epsilon+|Z_s|).$ 
	It follows that 
	\begin{align}    Y_s=\Phi(X_T)+\int_s^T(a_u^{\epsilon}Y_u+b_u^{\epsilon}Z_u+m_u^{\epsilon}+g(u,X_u,0,0))\,d\langle B \rangle _u-\int_s^TZ_u\,dB_u-(K_T-K_s)\,.
	\end{align}
	By \cite[Lemma~3.6]{hu2018quadratic}, the process $b^{\epsilon}$
is a $G$-BMO martingale generator. Moreover,  $a^{\epsilon}$ and $b^{\epsilon}$ belong to $\mathbb{H}^p(0,T)$. Let $\hat{\mathbb{E}}^{b^{\epsilon}}$ be the sublinear expectation induced by $b^{\epsilon}$. By \eqref{new_GBM}, the process $B^{\epsilon}:=B-\int_0^\cdot  b_u^{\epsilon}\,d\langle B \rangle_u$  is a $G$-Brownian motion under  $\hat{\mathbb{E}}^{b^{\epsilon}}.$ We have
	\begin{align}
	Y_s=\Phi(X_T)+\int_s^T(a_u^{\epsilon}Y_u+m_u^{\epsilon}+g(u,X_u,0,0))\,d\langle B^{\epsilon} \rangle _u-\int_s^TZ_udB_u^{\epsilon}-(K_T-K_s)\,.
	\end{align}
	Since $K\in \mathbb{L}_G^p(\Omega_T)$ for all $p\geq1$, by Lemma~\ref{K is MG in BMO}, 
	the process $K$ is a decreasing $G$-martingale under  $\hat{\mathbb{E}}^{b^{\epsilon}}$. Let  $X^{\epsilon}$ be a solution to the SDE
	\begin{align}
	dX_s^{\epsilon}=a_s^{\epsilon}X_s^{\epsilon}\,d\langle B^{\epsilon} \rangle _s \,,\;X_0=1\,.
	\end{align}
	It is easy to show that $X_s^{\epsilon}$ is bounded.
	Applying It\^{o}'s formula to $X^{\epsilon}Y$, we have
	\begin{align}
	Y_s&=(X_s^{\epsilon})^{-1}\tilde{\mathbb{E}}_s^{b^{\epsilon}}\Big[X_T^{\epsilon}\Phi(X_T)+\int_s^TX_u^{\epsilon}(m_u^{\epsilon}+g(u,X_u,0,0))\,d\langle B^{\epsilon} \rangle_u\Big]\,.
	\end{align} 
    Since $X_s^{\epsilon}, \Phi(X_T),$ and $g(s,X_s,0,0)$ are bounded
by a constant depending only on $C_1$, $C_2$, and $T$, under
$\tilde{\mathbb{E}}^{b^{\epsilon}}$ we have
 \begin{align}
     Y_s&\leq C(1+\tilde{\mathbb{E}}_s^{b^{\epsilon}}[\int_s^T(4C_1\epsilon(1+\epsilon+|Z_u|)\,d\langle B^\epsilon \rangle _u]   )=C(1+\epsilon+\hat{\mathbb{E}}_s^{b^{\epsilon}}[\int_s^T\epsilon|Z_u|\,du])\\
     &\leq
     C(1+\epsilon+\epsilon \hat{\mathbb{E}}_s[\frac{M_T^{b^\epsilon}}{M_s^{b^\epsilon}}\int_s^T|Z_u|\,du])\,.
 \end{align}
 for some constant $C>0$ depending only on $C_1$, $C_2$, $\overline{\sigma}$ and $T$. Since $|b^\epsilon|\leq C_1(1+|Z|)$,
$|\!|b^\epsilon|\!|_{BMO_G}$ is bounded uniformly in $\epsilon$.
Hence Lemma~\ref{reverse holder inequality} applies with a constant
independent of $\epsilon$.
 Applying H\"older's inequality and Lemma~\ref{reverse holder inequality},
using $\hat{\mathbb{E}}_s[(\int_s^T |Z_u|^2\,d\langle B \rangle _u)^p ]<\infty$,
and letting $\epsilon\rightarrow 0$, we conclude that
$\big|\!\big|\sup_{0\leq u \leq T}|Y_u|\big|\!\big|_{\mathbb{L}^{\infty}}$
is bounded by a constant depending only on $C_1$, $C_2$,
$\overline{\sigma}$, and $T$. 
\end{proof}

The following theorem establishes a comparison principle for
quadratic $G$-BSDEs.

\begin{thm}\label{in QBSDE comparison theorem}
    Suppose Assumption~\ref{assumption} holds and, for $k=1,2$,
let the triple
$(Y^k,Z^k,K^k)= (Y^{t,x,k},Z^{t,x,k},K^{t,x,k})$
be a solution to the following $G$-BSDE:
    \begin{align}
    Y_t^k=\Phi^k(X_T^{t,x})&+\int_t^Tf^k(u,X_u^{t,x},Y_u^k,Z_u^k)\,du+\int_t^Tg^k_{ij}(u,X_u^{t,x},Y_u^k,Z_u^k)\,d\langle B^i,B^j\rangle_u\\
    &-\int_t^TZ_u^k\,dB_u-(K_T^k-K_t^k)
    \end{align}
    where $f^k,g^k_{ij}$ and $\Phi^k$ satisfy Assumption~\ref{assumption2}. If $\Phi(x)^1\geq \Phi^2(x)$, $f^1(s,x,y,z)\geq f^2(s,x,y,z)$, $g^1_{ij}(s,x,y,z)\geq g^2_{ij}(s,x,y,z)$ for $s\in[t,T], x\in \mathbb{R}^m, y\in \mathbb{R}, z\in \mathbb{R}^d$ and $Y^k$ is bounded for each $k=1,2$, then $Y^1_s\geq Y^2_s$ for $(s,x)\in [t,T]\times \mathbb{R}^m$.
\end{thm}
\begin{proof}
For simplicity, assume $d=1\,,m=1\,, t=0$ and $f = 0$. 
    First, Theorem~\ref{Z is BMO and K is intgrable} implies that $Z_1$ and $Z_2$ are $G$-BMO martingale generators. Let $(\hat{Y},\hat{Z},\hat{K})=(Y^1-Y^2,Z^1-Z^2,K^1-K^2)$ and $\hat{\Phi}(X_T)=\Phi^1(X_T)-\Phi^2(X_T)$. For $\epsilon>0$, let $l^\epsilon(x)$ be a Lipschitz function satisfying $\mathds{1}_{[-\epsilon,\epsilon]}(x)\leq l^\epsilon(x)\leq \mathds{1}_{[-2\epsilon,2\epsilon]}(x)$. Define four processes
\begin{align}
    a_s^{\epsilon}&=(1-l^\epsilon(\hat{Y}_s)\frac{g^1(s,X_s,Y^1_s,Z^1_s)-g^1(s,X_s,Y_s^2,Z^1_s)}{\hat{Y}_s}\mathds{1}_{|\hat{Y}_s|\geq 0}\,,\\
    b_s^{\epsilon}&=(1-l^\epsilon(\hat{Z}_s)\frac{g^1(s,X_s,Y^2_s,Z^1_s)-g^1(s,X_s,Y_s^2,Z_s^2)}{|\hat{Z}_s|^2}\hat{Z}_s\mathds{1}_{|\hat{Z}_s|\geq 0}\,,\\
    m_s^{\epsilon}&=l^\epsilon(\hat{Y}_s)(g^1(s,X_s,Y^1_s,Z^1_s)-g^1(s,X_s,Y_s^2,Z^1_s))+l^\epsilon(\hat{Z}_s)(g^1(s,X_s,Y_s^2,Z^1_s)-g^1(s,X_s,Y_s^2,Z_s^2))\,,\\
    h_s&=g^1(s,X_s,Y_s^2,Z_s^2)-g^2(s,X_s,Y_s^2,Z_s^2)\,.
\end{align}
Then $|a_s^{\epsilon}|\leq C_1$, $|b_s^{\epsilon}|\leq C_1(1+|Z^1_s|+|Z^2_s|)$ and $m_s^{\epsilon}\leq 4\epsilon C_1(1+\epsilon+|Z^1_s|)$. It follows that
\begin{align}
\hat{Y_s}=\hat{\Phi}(X_T)+\int_s^T(a_u^{\epsilon}\hat{Y_u}+b_u^{\epsilon}\hat{Z}_u+m_u^{\epsilon}+h_u)\,d\langle B\rangle _u -\int_s^T\hat{Z}_ud\,B_u-(K^1_T-K^1_s)+(K^2_T-K^2_s)\,.
\end{align}
By \cite[Lemma 3.6]{hu2018quadratic}, $b^{\epsilon}$ is a $G$-BMO martingale generator. Then we have
\begin{align}
(\hat{Y_s}+K^2_s)=(\hat{\Phi}(X_T)+K^2_T)+\int_s^T(a_u^{\epsilon}\hat{Y_u}+m_u^{\epsilon}+h_u)\,d\langle B\rangle _u -\int_s^T\hat{Z}_u\,(dB_u-b_u^{\epsilon}\,d\langle B \rangle_u)-(K^1_T-K^1_s)\,.
\end{align}
 Let $\hat{\mathbb{E}}^{b^{\epsilon}}$ be the sublinear expectation induced by $b_s^{\epsilon}$. By \eqref{new_GBM}, $B^{\epsilon}:=B-\int_0^{\cdot} b_u^{\epsilon}\,d\langle B \rangle _u$ is a $G$-Brownian motion under $\hat{\mathbb{E}}^{b^\epsilon}$. 
 Let $X^{\epsilon}$ be a solution to the  SDE
\begin{align}
    dX_s^{\epsilon}=a_s^{\epsilon}X_s^{\epsilon}\,d\langle B^{\epsilon}\rangle _s\hspace{1cm} X^{\epsilon}_0=1\,.
\end{align}
Applying It\^{o}'s formula to $(\hat{Y}-K^2)X^{\epsilon}$ under $\hat{\mathbb{E}}^{b^{\epsilon}}$, we have
\begin{align}\label{in the proof of theorem 3.8 eqn}
X_s^{\epsilon}(\hat{Y_s}+K^2_s)=&X_T^{\epsilon}(\hat{\Phi}(X_T)+K^2_T)+\int_s^TX_u^{\epsilon}( m_u^{\epsilon}+h_u-a_u^{\epsilon}K_u^2)\,d\langle B^{\epsilon}\rangle _u\\&-\int_s^TX_u^{\epsilon}\hat{Z}_u\,dB_u^{\epsilon}-\int_s^TX_u^{\epsilon}\,dK_u^1\,.
\end{align}
Consider the following $G$-BSDE under $\hat{\mathbb{E}}^{b^\epsilon}$:
\begin{align}
    Y_s=K_T^2+\int_s^T(a_u^{\epsilon} Y_u-a_u^\epsilon K_u^2)\,d\langle B^{\epsilon} \rangle _u-\int_s^TZ_u\,dB_u^{\epsilon} -(K_T-K_s)\,.
\end{align}
It follows that $(K^2,0,K^2)$ is a solution. By \cite[Theorem~4.1]{hu2014backward} and
\cite[Theorem~3.2]{hu2014comparison}, we obtain 
\begin{align}\label{E.q in comparison thm }
    X_s^{\epsilon}K_s^2=\tilde{\mathbb{E}}^{b^{\epsilon}}_s[X_T^{\epsilon}K_T^2-\int_s^TX_u^{\epsilon}a_u^{\epsilon}K_u^2\,d\langle B^{\epsilon}\rangle_u].
\end{align}
Applying $\tilde{\mathbb{E}}_s^{b^{\epsilon}}$ to both sides of
\eqref{in the proof of theorem 3.8 eqn} and using
Lemma~\ref{K is MG in BMO} and
Theorem~\ref{Z is BMO and K is intgrable}, we obtain
\begin{align}
    X_s^{\epsilon}(\hat{Y_s}+K^2_s)=&\tilde{\mathbb{E}}_s^{b^{\epsilon}}[X_T^{\epsilon}(\hat{\Phi(X_T)}+K^2_T)+\int_s^TX_u^{\epsilon}( m_u^{\epsilon}+h_u-a_u^{\epsilon}K_u^2)\,d\langle B^{\epsilon}\rangle _u]\\
    \geq&\tilde{\mathbb{E}}_s^{b^{\epsilon}}[X_T^{\epsilon}K^2_T+\int_s^TX_u^{\epsilon}( m_u^{\epsilon}-a_u^{\epsilon}K_u^2)\,d\langle B^{\epsilon}\rangle _u]\,.
    \end{align}
By \eqref{E.q in comparison thm }, we have
\begin{align}
     X_s^{\epsilon}\hat{Y_s}\geq-\tilde{\mathbb{E}}_s^{b^{\epsilon}}[\int_s^T -X_u^{\epsilon} m_u^{\epsilon}\,d\langle B^{\epsilon}\rangle _u]\,.
\end{align}
    Letting $\epsilon\rightarrow 0$, we obtain $\hat{Y}\geq 0$.   
\end{proof}

\begin{thm}\label{thm:uniqueness}
Suppose Assumptions~\ref{assumption} and~\ref{assumption2} hold. Then, the BSDE \eqref{QBSDE} has at most one solution
$(Y,Z,K)$ such that $Y$ is bounded.
\end{thm}

\begin{proof}
For simplicity, we assume that $d=1$ and $t=0$. Let
$(Y^1,Z^1,K^1)$ and $(Y^2,Z^2,K^2)$ be two bounded solutions to
\eqref{QBSDE}. Applying
Theorem~\ref{in QBSDE comparison theorem} yields for $s\in [0,T]$, 
\begin{align}
Y_s^1=Y_s^2
\quad\text{quasi-surely. }
\end{align}
Hence, for every $P\in\mathcal{P}$, the processes $Y^1$ and $Y^2$
are modifications of each other under $P$. Since both processes have
continuous paths $P$-almost surely, they are indistinguishable under
$P$. Therefore, for every $p\geq1$,
\begin{align}
E^P\left[
\sup_{0\leq u\leq T}|Y_u^1-Y_u^2|^p
\right]=0,
\qquad P\in\mathcal{P}.
\end{align}
Taking the supremum over $P\in\mathcal{P}$, we obtain
\begin{align}
\hat{\mathbb{E}}\left[
\sup_{0\leq u\leq T}|Y_u^1-Y_u^2|^p
\right]=0.
\label{eqn:uni}
\end{align}
By \cite[Proposition 3.7]{hu2018quadratic} and \eqref{eqn:uni}, we obtain
\begin{align}
\hat{\mathbb{E}}\left[
\int_0^T|Z_u^1-Z_u^2|^2\,du
\right]=0.
\end{align}
Thus, $Z^1=Z^2$ in $\mathbb{H}^2(0,T;\mathbb{R})$. Finally, we have
$K^1=K^2$. Thus, \eqref{QBSDE} has at most one
bounded solution. 
\end{proof}

We now establish regularity and growth estimates with respect to
the initial condition $\xi$.
\begin{prop}\label{in QBSDE solution is Lipschitz and linear growth}
    Suppose Assumptions~\ref{assumption} and~\ref{assumption2} hold,
and let $(Y^{t,\xi},Z^{t,\xi},K^{t,\xi})$ and
$(Y^{t,\xi'},Z^{t,\xi'},K^{t,\xi'})$ be solutions to \eqref{QBSDE}. If $Y^{t,\xi},Y_s^{t,\xi'}$ are bounded, then there exists a constant $C$ depending only on $C_1,$ $\overline{\sigma}$, $T$ and the order $q$ of the reverse H\"older inequality for $M_Z$ such that
    \begin{align}
          |&Y_t^{t,\xi}-Y_t^{t,\xi'}|\leq  C|\xi-\xi'|\,,\\
        |&Y_t^{t,\xi}|\leq C(1+|\xi|)
    \end{align}
    for all $\xi, \xi'\in \mathbb{L}^1(\Omega_t;\mathbb{R}^m)$.
    \end{prop} 
    \begin{proof}
    The proof is similar to that of
Theorem~\ref{in QBSDE comparison theorem}, so we provide only
an outline.
    For simplicity, we assume $d=1,m=1$ and $f=0$. By Theorem~\ref{Z is BMO and K is intgrable}, $Z^{t,\xi}$ and $ Z^{t,\xi'} $ are $G$-BMO martingale generators. Let $(\hat{Y},\hat{Z},\hat{K})=(Y^{t,\xi}-Y^{t,\xi'},Z^{t,\xi}-Z^{t,\xi'},K^{t,\xi}-K^{t,\xi'})$ and $\hat{\Phi}(X_T)=\Phi(X_T^{t,\xi})-\Phi(X_T^{t,\xi'})$. As in the proof of Theorem~\ref{in QBSDE comparison theorem}, for $\epsilon>0$,
    there exist $a_s^{\epsilon},b_s^{\epsilon},m_s^{\epsilon}$ and $h_s$ such that
\begin{align}
\hat{Y_t}=\hat{\Phi}(X_T)+\int_t^T(a_u^{\epsilon}\hat{Y_u}+b_u^{\epsilon}\hat{Z}_u+m_u^{\epsilon}+h_u)\,d\langle B\rangle _u -\int_t^T\hat{Z}_u\,dB_u-(K^{t,\xi}_T-K^{t,\xi}_t)+(K^{t,\xi'}_T-K^{t,\xi'}_t)\,.
\end{align} 
Let $\hat{\mathbb{E}}^{b^{\epsilon}}$ be the sublinear expectation induced by $b^{\epsilon}$. By \eqref{new_GBM}, the process
$B^\epsilon=B-\int_0^\cdot b_u^\epsilon\,d\langle B
\rangle _u$ is a $G$-Brownian motion under
$\hat{\mathbb{E}}^{b^{\epsilon}}$. Define a process $X^{\epsilon}$ similarly. Then we have
\begin{align}
    X_t^{\epsilon}(\hat{Y_t}+K^{t,\xi'}_t)=&\hat{\mathbb{E}}_t^{b^{\epsilon}}[X_T^{\epsilon}(\hat{\Phi}(X_T)+K^{t,\xi'}_T)+\int_t^TX_u^{\epsilon}( m_u^{\epsilon}+h_u-a_u^{\epsilon}K_u^{t,\xi'})\,d\langle B^{\epsilon}\rangle _u]\\
    \leq&\hat{\mathbb{E}}_t^{b^{\epsilon}}[X_T^{\epsilon}K^{t,\xi'}_T-\int_t^TX_u^{\epsilon}a_u^{\epsilon}K_u^{t,\xi'}\,d\langle B^{\epsilon}\rangle _u]\\
    &+\hat{\mathbb{E}}_t^{b^{\epsilon}}[X_T^\epsilon|\Phi(X_T)|+\int_t^TX_u^{\epsilon}(m_u^{\epsilon}+ |h_u|)\,d\langle B^{\epsilon} \rangle _u]\,.
    \end{align}
Similarly
\begin{align}
    X_t^{\epsilon}K_t^{t,\xi'}=\hat{\mathbb{E}}^{b^{\epsilon}}_t[X_T^{\epsilon}K_T^{t,\xi'}-\int_t^TX_u^{\epsilon}a_u^{\epsilon}K_u^{t,\xi'}\,d\langle B^{\epsilon}\rangle_u].
\end{align}
Then we have 
    \begin{align}
     \hat{Y_t}
    \leq&(X_t^{\epsilon})^{-1} \hat{\mathbb{E}}_t^{b^{\epsilon}}[X_T^\epsilon|\Phi(X_T)|+\int_t^TX_u^{\epsilon}(m_u^{\epsilon}+ |h_u|)\,d\langle B^{\epsilon} \rangle _u]\\
    \leq& C\hat{\mathbb{E}}_t^{b^{\epsilon}}[|X_T^{t,\xi}-X_T^{t,\xi'}| +\int_t^T(m_u^{\epsilon}+ |X_u^{t,\xi}-X_u^{t,\xi'}|)\,d\langle B^{\epsilon} \rangle _u]\,.
    \end{align}
    Letting $\epsilon \rightarrow 0$ and using Theorem~\ref{SDE sublinear property}, we obtain $Y_t^{t,\xi}-Y_t^{t,\xi'}\leq C|\xi-\xi'|$. 
    Moreover, in the same way, we have $Y_t^{t,\xi'}-Y_t^{t,\xi}\leq C|\xi-\xi'|$. Combining these inequalities gives $|Y_t^{t,\xi}-Y_t^{t,\xi'}|\leq C|\xi-\xi'|$.
Similarly, we obtain $Y_t^{t,\xi}\leq C(1+|\xi|)$.
    \end{proof}
    \begin{remark}\label{order of reverse holder inquality rmk}
        The order of the reverse H\"older inequality for $M^Z$ depends on
$|\!|Z|\!|_{BMO_G}$. Therefore, the constant $C$ in
Proposition~\ref{in QBSDE solution is Lipschitz and linear growth}
also depends on $C_1$ and $C_2$.
    \end{remark}
    Let $(Y^{t,x},Z^{t,x},K^{t,x})$ be a solution to \eqref{QBSDE} under Assumptions~\ref{assumption} and \ref{assumption2}. Define a function 
    \begin{align}
     u(t,x)=Y_t^{t,x},\hspace{1cm} (t,x)\in [0,T]\times\mathbb{R}^m.   
    \end{align}
    Since $b$, $h_{ij}$, $\sigma_j$, $\Phi$, $f$, and $g_{ij}$
are deterministic functions and
$\tilde{B}_\cdot:= B_{t+\cdot}-B_t$ is a $G$-Brownian motion,
it follows that $u(t,x)$ is deterministic. If the process $Y^{t,x}$ is bounded for all $x\in \mathbb{R}^m$, then by Proposition~\ref{in QBSDE solution is Lipschitz and linear growth}, we get 
\begin{align}\label{deterministc U is lipschitz}
    |u(t,x)-u(t,x')|\leq C|x-x'|\,,\\
    |u(t,x)|\leq C(1+|x|)\,.
\end{align}
The following result plays a key role in proving the existence of
solutions to the $G$-BSDE.
\begin{thm}\label{QSDE sol is represented by u(t,X_t)}
     Suppose Assumptions~\ref{assumption} and~\ref{assumption2} hold,
and let the triple $(Y,Z,K)$ = $(Y^{t,\xi},Z^{t,\xi},K^{t,\xi})$
be a solution to \eqref{QBSDE}. If the process $Y^{t,\xi}$ is bounded for all $\xi\in \mathbb{L}^1(\Omega_t;\mathbb{R}^m)$, then
     \begin{align}
         u(t,\xi)=Y_t^{t,\xi}.
     \end{align}
\end{thm}
\begin{proof}
For simplicity, we assume $d=1,m=1$ and $f=0$. We first prove the result for bounded $\xi$.
For $n\in \mathbb{N}$ and $\rho>0$, define the simple random variable
    \begin{align}
        \eta^n:=\sum_{i=-n}^n\frac{i\rho}{n}\mathds{1}_{[\frac{i\rho}{n},\frac{(i+1)\rho}{n}]}{(\xi)}\,.
    \end{align}
 Using Proposition~\ref{in QBSDE solution is Lipschitz and linear growth}, we have
 \begin{align}
     |Y_t^{t,\xi}-u(t,\eta^n)|&=|Y_t^{t,\xi}-\sum_{i=-n}^n u(t,\frac{i\rho}{n}) \mathds{1}_{[\frac{i\rho}{n},\frac{(i+1)\rho}{n}]}(\xi)|\leq \sum_{i=-n}^n| Y_t^{t,\xi}-Y_t^{t,\frac{i\rho}{N}}| \mathds{1}_{[\frac{i\rho}{n},\frac{(i+1)\rho}{n}]}(\xi)\\
     &\leq C \sum_{i=-n}^n|\xi-\frac{i\rho}{n}|\mathds{1}_{[\frac{i\rho}{n},\frac{(i+1)\rho}{n}]}(\xi)\leq C\frac{\rho}{n}\,.
 \end{align}
 Using \eqref{deterministc U is lipschitz} and the same argument, we obtain
 \begin{align}
     |u(t,\xi)-u(t,\eta^n)|\leq C\frac{\rho}{n}\,.
 \end{align}
Then we have $|Y_t^{t,\xi}-u(t,\xi)|\leq 2C \frac{\rho}{n}$. Letting $n\rightarrow \infty$, we obtain $Y_t^{t,\xi}=u(t,\xi)$ for bounded $\xi\in \mathbb{L}^1(\Omega_t)$.
For general $\xi \in \mathbb{L}^1(\Omega_t)$, there exists a
sequence of bounded random variables $\xi^n$ such that
$\lim_{n\rightarrow\infty}\hat{\mathbb{E}}_t[|\xi-\xi^n|]=0$.
Passing to the limit yields $Y_t^{t,\xi}=u(t,\xi)$.
    
\end{proof}

\begin{thm}\label{u is 1/2 holder contin wrt t}
     Suppose Assumptions~\ref{assumption} and~\ref{assumption2} hold,
and let the triple $(Y^{t,x},Z^{t,x},K^{t,x})$
be a solution to \eqref{QBSDE}. Let $u(t,x):=Y_t^{t,x}$. If for each $x\in \mathbb{R}^m$, the process $Y^{t,x}$ is bounded, then there exists a constant $C$ depending only on $C_1$, $C_2$, $\underline{\sigma}$, $\overline{\sigma}$ and $T$ such that 
     \begin{align}
         |u(t,x)-u(t+\delta,x)|\leq C(1+|x|)\delta^{\frac{1}{2}}
     \end{align}
     for all  $(\delta,x)\in[0,T-t]\times \mathbb{R}^m$.
\end{thm}
\begin{proof}
For simplicity, we assume $d=1,m=1$ and $f=0$.
We use an argument similar to that in the proof of
Theorem~\ref{Z is BMO and K is intgrable}.
For $\epsilon>0$, let $l^\epsilon(x)$ be a Lipschitz function satisfying $\mathds{1}_{[-\epsilon,\epsilon]}(x)\leq l^\epsilon(x)\leq \mathds{1}_{[-2\epsilon,2\epsilon]}(x)$. Define three processes 
\begin{align}
    a_s^{\epsilon}&=(1-l^\epsilon(Y_s^{t,x}))\frac{g(s,X_s^{t,x},Y_s^{t,x},Z_s^{t,x})-g(s,X_s^{t,x},0,Z_s^{t,x})}{Y_s^{t,x}}\mathds{1}_{|Y_s^{t,x}|\geq 0}\,,\\
    b_s^{\epsilon}&=(1-l^\epsilon(Z_s^{t,x}))\frac{g(s,X_s^{t,x},0,Z_s^{t,x})-g(s,X_s^{t,x},0,0)}{|Z_s^{t,x}|^2}Z_s^{t,x}\mathds{1}_{|Z_s^{t,x}|\geq 0}\,,\\
    m_s^{\epsilon}&=l^\epsilon(Y_s^{t,x})(g(s,X_s^{t,x},Y_s^{t,x},Z_s^{t,x})-g(s,X_s^{t,x},0,Z_s^{t,x}))\\
    &\quad+l^\epsilon(Z_s^{t,x})(g(s,X_s^{t,x},0,Z_s^{t,x})-g(s,X_s^{t,x},0,0))\,.
\end{align}
Then $|a_s^{\epsilon}|\leq C_1$, $|b_s^{\epsilon}|\leq C_1(1+|Z_s^{t,x}|)$ and $m_s^{\epsilon}\leq 4\epsilon C_1(1+\epsilon+|Z_s^{t,x}|)$\,.
It follows that
\begin{align}    Y_t^{t,x}=\Phi(X_T^{t,x})+\int_t^T(a_u^{\epsilon}Y_u^{t,x}&+b_u^{\epsilon}Z_u^{t,x}+m_u^{\epsilon}+g(u,X_u^{t,x},0,0))\,d\langle B \rangle _u\\
&-\int_t^TZ_u^{t,x}\,dB_u-(K_T^{t,x}-K_t^{t,x})\,.
\end{align}
 Let $\hat{\mathbb{E}}^{b^{\epsilon}}$ be the sublinear expectation induced by $b^{\epsilon}$. By \eqref{new_GBM}, $B^{\epsilon}:=B-\int_0^{\cdot} b_u^{\epsilon}\,d\langle B \rangle _u$ is a $G$-Brownian motion under $\hat{\mathbb{E}}^{b^\epsilon}$. We have 
\begin{align}
    Y_t^{t,x}=\Phi(X_T^{t,x})+\int_t^T(a_u^{\epsilon}Y_u^{t,x}+m_u^{\epsilon}+g(u,X_u^{t,x},0,0))\,d\langle B^{\epsilon} \rangle _u-\int_t^TZ_u^{t,x}\,dB_u
    ^{\epsilon}-(K_T^{t,x}-K_t^{t,x})\,.
\end{align}
 Since $K_T^{t,x}\in \mathbb{L}^p(\Omega_T)$ for all $p\geq 1$, by Lemma~\ref{K is MG in BMO},  the process $K$ is a decreasing $G$-martingale under $\hat{\mathbb{E}}^{b^{\epsilon}}$. Let  $X^{\epsilon}$ be a solution to the  SDE
\begin{align}
    dX_s^{\epsilon}=a_s^{\epsilon}X_s^{\epsilon}\,d\langle B^{\epsilon} \rangle _s \hspace{1cm} X_0=1\,.
\end{align}
Applying It\^{o}'s formula to $X^\epsilon Y$, we have
\begin{align}
    Y_t^{t,x}=(X_t^{\epsilon})^{-1}\hat{\mathbb{E}}_t^{b^{\epsilon}}[X_{t+\delta}^{\epsilon}Y_{t+\delta}^{t,x}+\int_t^{t+\delta}X_u^{\epsilon}(m_u^{\epsilon}+g(u,X_u^{t,x},0,0))\,d\langle B^{\epsilon} \rangle _u].
\end{align}
Since $Y_t^{t,x}$ is deterministic, for $\delta\in[0,T-t]$ we have
\begin{align}
u(t,x)=\hat{\mathbb{E}}^{b^{\epsilon}}[\frac{X_{t+\delta}^{\epsilon}}{X_t^{\epsilon}}Y_{t+\delta}^{t,x}+\int_t^{t+\delta}\frac{X_{u}^{\epsilon}}{X_t^{\epsilon}}(m_u^{\epsilon}+g(u,X_u^{t,x},0,0))\,d\langle B^{\epsilon}\rangle_u]\,.
\end{align}
From $X_s^{t,x}=X_s^{t+\delta, X_{t+\delta}^{t,x}}$ for $(s,x)\in [t+\delta,T]\times \mathbb{R}$ and Theorem~\ref{thm:uniqueness}, we have $Y_{t+\delta}^{t,x}=Y_{t+\delta}^{t+\delta,X_{t+\delta}^{t,x}}$. Then by Theorem~\ref{QSDE sol is represented by u(t,X_t)}, it follows that $Y_{t+\delta}^{t,x}=u(t+\delta, X_{t+\delta}^{t,x})$. It follows that 
\begin{align}
    |u(t,x)-u(t+\delta,x)|=&\Big|\hat{\mathbb{E}}^{b^{\epsilon}}\Big[\frac{X_{t+\delta}^{\epsilon}}{X_t^{\epsilon}}u(t+\delta,X_{t+\delta}^{t,x})-u(t+\delta,x)\\
    &+\int_t^{t+\delta}\frac{X_{u}^{\epsilon}}{X_t^{\epsilon}}(m_u^{\epsilon}+g(u,X_u^{t,x},0,0))\, d \langle B^{\epsilon}\rangle _u\Big]\Big|\\
    =&\Big|\hat{\mathbb{E}}^{b^{\epsilon}}\Big[\frac{X_{t+\delta}^{\epsilon}}{X_t^{\epsilon}}\big(u(t+\delta,X_{t+\delta}^{t,x})-u(t+\delta,x)\big)+u(t+\delta,x)(\frac{X_{t+\delta}^{\epsilon}}{X_t^{\epsilon}}-1)\\
    &+\int_t^{t+\delta}\frac{X_{u}^{\epsilon}}{X_t^{\epsilon}}(m_u^{\epsilon}+g(u,X_u,0,0))\,d\langle B^{\epsilon}\rangle _u    \Big]\Big|\,.
\end{align}
By Proposition~\ref{in QBSDE solution is Lipschitz and linear growth},
Theorem~\ref{SDE sublinear property}, and the boundedness of
$\frac{X_{t+\delta}^{\epsilon}}{X_t^{\epsilon}}$, we obtain
\begin{align}
 &\hat{\mathbb{E}}^{b^{\epsilon}}[|\frac{X_{t+\delta}^{\epsilon}}{X_t^{\epsilon}}(u(t+\delta,X_{t+\delta}^{t,x})-u(t+\delta,x))|]\leq C(1+|x|)\delta^{\frac{1}{2}}\,,\\
 &\hat{\mathbb{E}}^{b^\epsilon}[|u(t+\delta,x) (\frac{X_{t+\delta}^{\epsilon}}{X_t^{\epsilon}}-1)|]\leq C(1+|x|)\delta\,,\\
 &\hat{\mathbb{E}}^{b^\epsilon}[|\int_t^{t+\delta}\frac{X_{u}^{\epsilon}}{X_t^{\epsilon}}g(u,X_u,0,0)\,d\langle B^{\epsilon}\rangle _u|]\leq C\delta
\end{align}
where the constant $C$ depends on $C_1$, $C_2$, $\underline{\sigma}$, $\overline{\sigma}$ and $T$.
Letting $\epsilon \rightarrow 0$ yields the desired result.
    
\end{proof}

We now establish a stability result.
\begin{thm}\label{in QBSDE convergence stability}
     Suppose Assumption~\ref{assumption} holds and, for each $n\geq 1$,
let the triple $(Y^n,Z^n,K^n)= (Y^{t,x,n},Z^{t,x,n},K^{t,x,n})$
be a solution to the following $G$-BSDE:
    \begin{align}
    Y_s^{n,t,x}=&\Phi^n(X_T^{t,x})+\int_s^Tf^n(u,X_u^{t,x},Y_u^{n,t,x},Z_u^{n,t,x})\,du+\int_s^Tg^n_{ij}(u,X_u^{t,x},Y_u^{n,t,x},Z_u^{n,t,x})\,d\langle B^i,B^j\rangle_u\\
    &-\int_s^TZ_u^{n,t,x}\,dB_u-(K_T^{n,t,x}-K_s^{n,t,x})\,,\;t\le s\le T
    \end{align}
    where $f^n,g^n_{ij},\Phi^n$ satisfy Assumption~\ref{assumption2} uniformly. If the process $Y_s^{n,t,x}$ is bounded and $f^n \rightarrow f$, \,$g^n_{ij}\rightarrow g_{ij},\, \Phi^n\rightarrow \Phi$, then there exists a solution $(Y,Z,K)$ to the following equation
    \begin{align}
        Y_s^{t,x}=&\Phi(X_T^{t,x})+\int_s^Tf(u,X_u^{t,x},Y_u^{t,x},Z_u^{t,x})\,du+\int_s^Tg_{ij}(u,X_u^{t,x}Y_u^{t,x},Z_u^{t,x})\,d\langle B^i,B^j\rangle_u\\
        &-\int_s^TZ_u^{t,x}\,dB_u-(K_T^{t,x}-K_s^{t,x})\,,\;t\le s\le T\,.
    \end{align}
\end{thm}
\begin{proof}   For simplicity, assume $d=1,m=1$ and $f=0$.
We first show that, for some $p\geq 1$,
    \begin{align}
        \lim_{n.m\rightarrow \infty}\hat{\mathbb{E}}[\sup_{t\leq u \leq T}|Y_u^{n,t,x}-Y_u^{m,t,x}|^p]=0.
    \end{align}
    Let $u^n(s,x):=Y_s^{n,s,x}$. By Theorem~\ref{u is 1/2 holder contin wrt t}, Theorem~\ref{Z is BMO and K is intgrable} and \eqref{deterministc U is lipschitz}, we have 
    \begin{align}
        |u^n(s,x)-u^n(s+\delta, x')|&\leq C\big(|x-x'|+(1+|x|)\delta^{\frac{1}{2}}\big) \text{ for }(s,x)\in[t,T]\times \mathbb{R}\,,\\
        \sup_{(s,x)\in [t,T]\times \mathbb{R}}|u^n(s,x)|&\leq C\,,
    \end{align}
where the constant $C$ depends only on $C_1$, $C_2$,
$\underline{\sigma}$, $\overline{\sigma}$, and $T$, and is independent of the sequence index. By the Arzel\`a--Ascoli theorem, there exist a function $u(s,x)$
and a subsequence $u^{n_\ell}(s,x)$ such that $u^{n_\ell}(s,x)$
converges to $u(s,x)$ uniformly on every compact subset of
$[t,T]\times\mathbb{R}$.
For notational simplicity, we relabel $n_\ell$ as $n$.
It follows that
\begin{align}
    \lim_{n\rightarrow \infty}\hat{\mathbb{E}}[\sup_{t\leq s \leq T} |u^n(s,X^{t,x}_s)-u(s,X^{t,x}_s)|^p]\leq & \lim_{n\rightarrow \infty}\hat{\mathbb{E}}[\sup_{t\leq s \leq T} |u^n(s,X^{t,x}_s)-u(s,X^{t,x}_s)|^p\mathds{1}_{|X^{t,x}_s|\leq N}]\\
    &+\lim_{n\rightarrow \infty}\hat{\mathbb{E}}[\sup_{t\leq s \leq T} |u^n(s,X^{t,x}_s)-u(s,X^{t,x}_s)|^p\mathds{1}_{|X^{t,x}_s|\geq N}]\\
    &\leq (2C)^{p-1}\lim_{n\rightarrow \infty}\sup_{t\leq s\leq T,\, |x|\leq N}|u^n(s,x)-u(s,x)|\\
    &+(2C)^p\frac{\hat{\mathbb{E}}[\sup_{t\leq s\leq T}|X^{t,x}_s|]}{N} 
\end{align}
where $p\geq 1$, $N>0$ and the constant $C$ depends on $C_1$, $C_2$, $\underline{\sigma}$, $\overline{\sigma}$ and $T$.
Sending $N\rightarrow \infty$, we have \begin{align}\label{eqn:sta}
    \lim_{n\rightarrow \infty}\hat{\mathbb{E}}[\sup_{t\leq s \leq T} |u^n(s,X^{t,x}_s)-u(s,X^{t,x}_s)|^p]=0\,.
\end{align}
For $(s,x)\in [t,T]\times\mathbb{R}$, define the process
$Y^{t,x}_s=u(s,X^{t,x}_s)$.

Next, we show that, for every $p\geq 2$, 
\begin{align}
 \lim_{n,m \rightarrow \infty}\hat{\mathbb{E}}[(\int_t^T |Z_u^{n,t,x}-Z_u^{m,t,x}|^2\,d\langle B\rangle _u )^\frac{p}{2}]=0\,.   
\end{align}
 Define the triple $(\hat{Y},\hat{Z},\hat{K}):=(Y^{n,t,x}-Y^{m,t,x},Z^{n,t,x}-Z^{m,t,x},K^{n,t,x}-K^{m,t,x})$ and $\hat{\Phi}(X^{t,x}_T):=\Phi^n(X^{t,x}_T)-\Phi^m(X^{t,x}_T)$. Applying It\^{o}'s formula to $(\hat{Y}_t)^2$, we have 
\begin{align}
    \int_t^T |\hat{Z}_u|^2\,d\langle B \rangle _u\leq |\hat{\Phi}(X^{t,x}_T)|^2+2\int_t^T \hat{Y}_u&(g^n(u,X^{t,x}_u,Y^{n,t,x}_u,Z^{n,t,x}_u)-g^m(u,X^{t,x}_u,Y_u^{m,t,x},Z^{m,t,x}_u))\,d\langle B\rangle_u\\&-2\int_t^T \hat{Y}_u\hat{Z}_u\,dB_u-2\int_t^T \hat{Y}_u\,d\hat{K}_u\,.
\end{align}
 By \cite[Theorem 2.1]{gao2009pathwise}, we have 
 \begin{align}
     \hat{\mathbb{E}}[(\int_t^T |\hat{Z}_u|^2\,d\langle B \rangle _u)^\frac{p}{2}]\leq &C\big(\hat{\mathbb{E}}[{|\hat{\Phi}(X^{t,x}_T)|}^p]\\
     &+\hat{\mathbb{E}}[\sup_{t\leq u \leq T}{|\hat{Y}_u|}^{\frac{p}{2}}{(\int_t^T 1+|Y_u^{n,t,x}|+|Y_u^{m,t,x}|+{|Z_u^{n,t,x}|}^2+{|Z_u^{m,t,x}|}^2\,d\langle B\rangle_u)}^{\frac{p}{2}}]\\&+\hat{\mathbb{E}}[\sup_{t\leq u\leq T}|{\hat{Y}_u|}^{\frac{p}{2}}{(\int_t^T {|\hat{Z}_u|}^2\,d\langle B\rangle _u)}^{\frac{p}{4}}]+\hat{\mathbb{E}}[\sup_{t\leq u \leq T}{|\hat{Y}_u|}^{\frac{p}{2}}\, |\hat{K}_T|^{\frac{p}{2}}|]\big)
 \end{align}
where the constant $C$ depends on $C_1$, $C_2$, $\underline{\sigma}$, $\overline{\sigma}$.
 By Theorem~\ref{Z is BMO and K is intgrable}, for every $p\geq 1$, $\big|\!\big|\sup_{t\leq u \leq T}|Y_u^{n,t,x}|^p\big|\!\big|_{\mathbb{L}^{\infty}}$, $\hat{\mathbb{E}}[(\int_t^T |Z_u^{n,t,x}|^2\,d\langle B\rangle_u)^p]$ and $\hat{\mathbb{E}}[|K_T^{n,t,x}|^p]$ are bounded uniformly. From \eqref{eqn:sta}, it follows that  $\lim_{n,m\to \infty}\hat{\mathbb{E}}[(\int_t^T |Z_u^{n,t,x}-Z_u^{m,t,x}|^2\,d\langle B\rangle_u )^\frac{p}{2}]=0$. We then define the process
 \begin{align}
     Z_s^{t,x}=\lim_{n\rightarrow \infty}Z_s^{n,t,x}
 \end{align}
 Next, we establish the following convergence result: 
 \begin{align}
 \lim_{n\rightarrow \infty}\hat{\mathbb{E}}[\int_t^T|g^n(u,X^{t,x}_u,Y_u^{n,t,x},Z_u^{n,t,x})-g(u,X^{t,x}_u,Y^{t,x}_u,Z^{t,x}_u)|d\langle B \rangle_u]=0\,.
 \end{align}
A direct calculation gives
\begin{align}
    \hat{\mathbb{E}}[\int_t^T&|g^n(u,X^{t,x}_u,Y_u^{n,t,x},Z_u^{n,t,x})-g(u,X^{t,x}_u,Y^{t,x}_u,Z^{t,x}_u)|\, d\langle B\rangle_u]\\
     \leq &\hat{\mathbb{E}}[\int_t^T|g^n(u,X^{t,x}_u,Y^{t,x}_u,Z^{t,x}_u)-g(u,X^{t,x}_u,Y^{t,x}_u,Z^{t,x}_u)| \,d\langle B\rangle_u]\\&+C_1\hat{\mathbb{E}}[(\int_t^T |Y_u^{n,t,x}-Y^{t,x}_u|+(1+|Z_u^{n,t,x}|+|Z^{t,x}_u|)|Z_u^{n,t,x}-Z^{t,x}_u|\,d\langle B\rangle_u]\,.
\end{align}
We have $\int_t^T|g^n(u,X^{t,x}_u,Y^{t,x}_u,Z^{t,x}_u)-g(u,X^{t,x}_u,Y^{t,x}_u,Z^{t,x}_u)| \,d\langle B\rangle_u\in \mathbb{L}^1_G(\Omega_T)$, and this random variable converges to $0$ quasi-surely. Also $\int_t^T|g^n(u,X^{t,x}_u,Y^{t,x}_u,Z^{t,x}_u)-g(u,X^{t,x}_u,Y^{t,x}_u,Z^{t,x}_u)| \,d\langle B\rangle_u\leq \int_t^T C(1+|Y^{t,x}_u|+|Z^{t,x}_u|^{2})\,d\langle B\rangle_u \in \mathbb{L}^1_b(\Omega_T)$ where the constant $C$ depends on $C_1$ and $C_2$. Applying \cite[Theorem 3.2]{hu2019convergences} and H\"older's inequality, we have  \begin{align}
 \lim_{n\rightarrow \infty}\hat{\mathbb{E}}[(\int_t^T|g^n(u,X^{t,x}_u,Y_u^{n,t,x},Z_u^{n,t,x})-g(u,X^{t,x}_u,Y^{t,x}_u,Z^{t,x}_u)| \,d\langle B\rangle_u]=0\,.
 \end{align}
 We now show that $K_s^{n,t,x}$ forms a Cauchy sequence in
$\mathbb{L}^p(\Omega_s)$ for every $p\geq 1$ and $s\in[t,T]$.
Since $(Y^{n,t,x},Z^{n,t,x},K^{n,t,x})$ is a solution, we have 
  \begin{align}
    K_s^{n,t,x}=Y_s^{n,t,x}-Y_t^{n,t,x}+\int_t^sg^n(u,X^{t,x}_u,Y_u^{n,t,x},Z_u^{n,t,x})\,\langle B\rangle_u-\int_t^sZ_u^{n,t,x}\,dB_u
    \end{align}
Thus, $K_s^{n}$ forms a Cauchy sequence in
$\mathbb{L}^p(\Omega_s)$. Denote the limit by $K$.
Finally, we show that $K$ is a decreasing $G$-martingale.
For each $s\leq u$, we have 
\begin{align}
    \hat{\mathbb{E}}[|\hat{\mathbb{E}}_s[K_u]-K_s|]=\hat{\mathbb{E}}[|\hat{\mathbb{E}}_s[K_u]-\hat{\mathbb{E}}_s[K_u^{n,t,x}]+K^{n,t,x}_s-K_s|]
    \leq\hat{\mathbb{E}}[|K_u-K_u^n|]+\hat{\mathbb{E}}[|K_s^n-K_s|]
\end{align}
Thus, the triple $(Y,Z,K)$ is a solution.
    \end{proof}
    An argument similar to that used in the proof of
Theorem~\ref{in QBSDE convergence stability} yields the following
corollary.
    \begin{cor}\label{cor needed in Erogidc}
         Suppose Assumption~\ref{assumption} holds and, for each $n\geq 1$,
let the triple $(Y^n,Z^n,K^n)= (Y^{t,x,n},Z^{t,x,n},K^{t,x,n})$
be a solution to the following $G$-BSDE:
    \begin{align}
    Y_s^{t,x,n}=&\Phi^n(X_T^{t,x})+\int_s^Tf^n(u,X_u^{t,x},Y_u^{t,x,n},Z_u^{t,x,n})\,du+\int_s^Tg^n_{ij}(u,X_u^{t,x},Y_u^{t,x,n},Z_u^{t,x,n})\,d\langle B^i,B^j\rangle_u\\
    &-\int_s^TZ_u^{t,x,n}\,dB_u-(K_T^{t,x,n}-K_s^{t,x,n}) \text{   where } (s,x)\in[t,T]\times \mathbb{R}^m
    \end{align}
    where $f^n,g^n,\Phi^n$ satisfy Assumption~\ref{assumption2} uniformly. If $f^n \rightarrow f$, $g^n_{ij}\rightarrow g_{ij}, \Phi^n\rightarrow \Phi$ and $|\!|Y^n-Y^m|\!|_{\mathbb{S}^2}$ tends to zero as both indices tend to infinity and $|\!|Z^n_s|\!|^2_{BMO_G}$ is bounded uniformly, then there exists a solution $(Y,Z,K)$ to the following equation
    \begin{align}
        Y_s^{t,x}=&\Phi(X_T^{t,x})+\int_s^Tf(u,X_u^{t,x},Y_u^{t,x},Z_u^{t,x})\,du+\int_s^Tg_{ij}(u,X_u^{t,x},Y_u^{t,x},Z_u^{t,x})\,d\langle B^i,B^j\rangle_u\\
        &-\int_s^TZ_u^{t,x}\,dB_u-(K_T^{t,x}-K_s^{t,x}) 
    \end{align}
    \end{cor}
    
\subsection{Existence of solutions}
\label{sec:exist}

We now state the main result of this section.
\begin{thm}\label{thm:exist}
    Suppose Assumptions~\ref{assumption} and \ref{assumption2} hold. Then \eqref{QBSDE} has a unique solution $(Y^{t,x},Z^{t,x},K^{t,x})$ such that $Y^{t,x}$ is a bounded process.
\end{thm}
\begin{proof}
For simplicity, assume $d=1,m=1,t=0$ and $f=0$.
Uniqueness follows from Theorem~\ref{thm:uniqueness}: if a triple
$(Y,Z,K)$ is a solution to \eqref{QBSDE} and $Y$ is bounded,
then it is the unique solution in this class.

We now prove existence. For each $n,\ell\geq1 $, define the following
two functions:
    \begin{align}
        g^n=g\wedge n,\hspace{0.5cm}
        g^{n,\ell}=g^n \vee (-\ell),.
    \end{align}    
Define the nonincreasing Lipschitz function
$\iota:[0,\infty)\rightarrow[0,1]$ by
\begin{align}
    \iota(r)
    :=
    \begin{cases}
        1, & 0\leq r\leq 1,\\
        2-r, & 1<r<2,\\
        0, & r\geq 2.
    \end{cases}
\end{align}
For each $k\ge 1$, let $\iota_k(z):=\iota\left(\frac{|z|}{k}\right),
    \qquad z\in\mathbb{R}^{d}\,.$ Then $\iota_k(z)$ is bounded and Lipschitz continuous in its argument
and nondecreasing with respect to $k$.
Define
\begin{align}
    g^{n,\ell,k}(s,x,y,z)
    &:=
    -(n+\ell)
    +
    \iota_k(z)
    \left(
        {g}^{n,\ell}(s,x,y,z)
        +
        (n+\ell)
    \right).
    \label{eq:cutoff_generator}
\end{align}
It follows that $g^n \uparrow g$, $g^{n,\ell}\downarrow g^n$ and $g^{n,\ell,k}\uparrow g^{n,\ell}$, as $n,\ell,k\rightarrow \infty$. Moreover, we have $g^n\leq n$, $-\ell\leq g^{n,\ell}\leq n,\;-(n+\ell)\leq g^{n,\ell,k}\leq n$ and $g^n , g^{n,\ell}$ satisfy Assumption~\ref{assumption2} uniformly. A direct calculation also gives, for every $x,x',y,y',z,z'\in \mathbb{R}$, 
\begin{align}
    &
    |
    {g}^{n,\ell,k}(s,x,y,z)
    -
    {g}^{n,\ell,k}(s,x',y',z')
    |\leq
    C_1\bigl(|x-x'|+|y-y'|\bigr)
    +
    \left(
        C_1(1+4k)
        +
        \frac{2(n+\ell)}{k}
    \right)|z-z'|.
\end{align}
    Then by \cite[Theorem 4.1]{hu2014backward}, for fixed $n,\ell,k\in \mathbb{N}$, the following $G$-BSDE has a unique solution $(Y^{n,\ell,k},Z^{n,\ell,k},K^{n,\ell,k})$
    \begin{align}
    Y_s^{n,\ell,k}=\Phi(X_T)+\int_s^Tg^{n,\ell,k}(u,X_u,Y_u^{n,\ell,k},Z_u^{n,\ell,k})\,d\langle B\rangle_u-\int_s^TZ_u^{n,\ell,k}\,dB_u-(K^{n,\ell,k}_T-K^{n,\ell,k}_s)\,.
    \end{align}
    It follows that  
    \begin{align}\label{E.q in the proof of existence}
        |Y_s^{n,\ell,k}|\leq C_2+(n+\ell)\overline{\sigma}T\,.
    \end{align}
    For fixed $n,\ell\in \mathbb{N}$ and $k>l$, let $(\hat{Y},\hat{Z},\hat{K}):=(Y^{n,\ell,k}-Y^{n,\ell,l},Z^{n,\ell,k}-Z^{n,\ell,l},K^{n,\ell,k}-K^{n,\ell,l})$. 
   By \cite[Theorem 3.6]{hu2014comparison}, the random variable $Y_s^{n,\ell,k}$ is increasing in $k$ for $s\in [0,T]$. Thus, we define the process $Y^{n,\ell}$ by
   \begin{align}
       Y^{n,\ell}=\lim_{k\rightarrow \infty}Y^{n,\ell,k}\quad \textnormal{quasi-surely}\,. 
   \end{align}
Applying It\^{o}'s formula to $|\hat{Y}|^2$, we have
    \begin{align}
        |\hat{Y}_s|^2=&\int_s^T2\hat{Y}_u(g^{n,\ell,k}(u,X_u,Y_u^{n,\ell,k},Z_u^{n,\ell,k})-g^{n,\ell,l}(u,X_u,Y_u^{n,\ell,l},Z_u^{n,\ell,l}))\,d\langle \,B \rangle_u\\&-\int_s^T|\hat{Z}_u|^2\,d\langle B\rangle_u-\int_s^T2\hat{Y}_u\hat{Z}_u\,dB_u-\int_s^T2\hat{Y}_u\,dK^{n,\ell,k}_u+\int_s^T2\hat{Y}_u\,dK^{n,\ell,l}_u \\
        \leq&
        \int_s^T2\hat{Y}_u(g^{n,\ell,k}(u,X_u,Y_u^{n,\ell,k},Z_u^{n,\ell,k})-g^{n,\ell,l}(u,X_u,Y_u^{n,\ell,l},Z_u^{n,\ell,l}))\,d\langle \,B \rangle_u\\&-\int_s^T2\hat{Y}_u\hat{Z}_u\,dB_u-\int_s^T2\hat{Y}_u\,dK^{n,\ell,k}_u\,.
    \end{align}
    Taking the conditional sublinear expectation $\hat{\mathbb{E}}_s$
on both sides, we obtain
    \begin{align}
        |\hat{Y}_s|^2 \leq &\hat{\mathbb{E}}_s[\int_s^T2|\hat{Y}_u|\,|(g^{n,\ell,k}(u,X_u,Y_u^{n,\ell,k},Z_u^{n,\ell,k})-g^{n,\ell,l}(u,X_u,Y_u^{n,\ell,l},Z_u^{n,\ell,l}))|\,d\langle B \rangle_u]\\
        &\leq \hat{\mathbb{E}}_s[\int_0^T4(\ell+n)\overline{\sigma}|\hat{Y}_u|\,du]\,.
    \end{align}
     By \cite[Theorem 2.4]{kim2024g}, we have for some $1\leq \gamma\leq 2$
    \begin{align}
       \hat{\mathbb{E}}[\sup_{0\leq s \leq T}|\hat{Y}_s|^2] &\leq \hat{\mathbb{E}}[\sup_{0\leq s \leq T}\hat{\mathbb{E}}_s[\int_0^T4(\ell+n)\overline{\sigma}|\hat{Y}_u|\,du]]\\
       &\leq C\{(\hat{\mathbb{E}}[(\int_0^T4(\ell+n)\overline{\sigma}|\hat{Y}_u|\,du)^2])^{\frac{1}{2}}+\hat{\mathbb{E}}[(\int_0^T4(\ell+n)\overline{\sigma}|\hat{Y}_u|\,du)^2]^\frac{1}{\gamma}\}\,.
    \end{align}
    Then by \cite[Theorem 3.2]{hu2019convergences} and \eqref{E.q in the proof of existence}, we have
    \begin{align}
       \lim_{k\rightarrow \infty} \hat{\mathbb{E}}[\sup_{0\leq s\leq T}|Y_s^{n,\ell}-Y_s^{n,\ell,k}|^2]=0\,.
    \end{align}
    We now show that, for each $k\geq 1$, $Z^{n,\ell,k}$ is a
$G$-BMO martingale generator. Applying It\^{o}'s formula to
$e^{-Y_s^{n,\ell,k}}$, we obtain
    \begin{align}
         \int_{\tau}^T \frac{1}{2}e^{- Y_u^{n,\ell,k}}|Z_u^{n,\ell,k}|^2\,d\langle B \rangle _u\leq e^{- \Phi(X_T)}&+\int_{\tau}^T- e^{-Y_u^{n,\ell,k}}g^{n,\ell,k}(u,X_u,Y_u^{n,\ell,k},Z_u^{n,\ell,k})\,d\langle B \rangle _u\\&+\int_{\tau}^Te^{-Y_u^{n,\ell,k}}Z_u^{n,\ell,k}\,dB_u\,.
    \end{align}
     Since $|g^{n,\ell,k}|\leq \ell+n$, the norm
$|\!|Z^{n,\ell,k}|\!|_{BMO_G}^2$ is bounded by a constant
independent of $k$. 
     By an argument similar to that used in the proof of
Theorem~\ref{Z is BMO and K is intgrable}, $\hat{\mathbb{E}}[|K^{n,\ell,k}|^p]$ is uniformly bounded.
    An argument similar to that used in the proof of
Theorem~\ref{in QBSDE convergence stability} yields the following
convergence results:
    \begin{align}
     &\lim_{k,l\rightarrow\infty}\hat{\mathbb{E}}[\int_0^T |Z_u^{n,\ell,k}-Z_u^{n,\ell,l}|^2\,d\langle B\rangle_u]=0\,,\\
     &\lim_{k,l\rightarrow\infty}\hat{\mathbb{E}}[(\int_0^T|g^{n,\ell,k}(u,X_u,Y_u^{n,\ell,k},Z_u^{n,\ell,k})-g^{n,\ell,l}(u,X_u,Y_u^{n,\ell,l},Z_u^{n,\ell,l})|\, d\langle B\rangle_u]=0\,,\\
     &\lim_{k,l\rightarrow\infty}\hat{\mathbb{E}}[|K_T^{n,\ell,k}-K_T^{n,\ell,l}]=0\,.
    \end{align}
These convergence results allow us to define the processes
$Z^{n,\ell}$ and $K^{n,\ell}$ by 
\begin{align}
    &Z^{n,\ell}=\lim_{k\rightarrow \infty} Z^{n,\ell,k}\,,\\
    &K^{n,\ell}=\lim_{k\rightarrow \infty}K^{n,\ell,k}\,.
\end{align}
It follows that the triple $(Y^{n,\ell},Z^{n,\ell},K^{n,\ell})$
is a solution to the following $G$-BSDE:
\begin{align}
    Y_s^{n,\ell}&=\Phi(X_T)+\int_s^Tg^{n,\ell}(u,X_u,Y_u^{n,\ell},Z_u^{n,\ell})\,d\langle B\rangle_u-\int_s^TZ_u^{n,\ell}\,dB_u-(K_T^{n,\ell}-K_s^{n,\ell}) \,.
\end{align}
By \eqref{E.q in the proof of existence}, the constructed process
$Y^{n,\ell}$ is bounded. Since $f^{n,\ell}$ and $g^{n,\ell}$
satisfy Assumption~\ref{assumption2}, for fixed $n$ we apply
Theorem~\ref{in QBSDE convergence stability} to obtain a triple
$(Y^n,Z^n,K^n)$ solving the following $G$-BSDE:
\begin{align}
     Y_s^{n}&=\Phi(X_T)+\int_s^Tg^{n}(u,X_u,Y_u^{n},Z_u^{n})\,d\langle B\rangle_u-\int_s^TZ_u^{n}\,dB_u-(K_T^{n}-K_s^{n}) \,.
\end{align}
 Finally, Theorem~\ref{in QBSDE convergence stability} yields a
solution to \eqref{QBSDE}.  
\end{proof}
\begin{thm}
    Suppose Assumptions~\ref{assumption} and \ref{assumption2} hold. Let $(Y^{t,x},Z^{t,x},K^{t,x})$ be the solution to \eqref{QBSDE}. Define $u(t,x)=Y_t^{t,x}$. Then $u(t,x)$ is a viscosity solution to the following PDE:
    \begin{align}
        \begin{cases}
 \partial_t u+F(D^2_xu, D_xu,u,x,t)=0 \\
u(T,x)=\Phi(x)
\end{cases}
\end{align}
where 
\begin{align}
F(D^2_xu, D_xu,u,x,t)=&G(H(D^2_xu,D_xu,u,x,t))+\langle b(t,x),D_xu\rangle\\
&+f(t,x,u,\langle\sigma_1(t,x),D_xu\rangle,\cdots,\langle\sigma_d(t,x),D_xu\rangle)\\
H_{ij}(D^2_xu,D_xu,t,x,t)=&\langle D^2_xu\sigma_i(t,x),\sigma_j(t,x)\rangle+2\langle D_xu,h_{ij}(t,x)\rangle\\
&+2\langle g_{ij}(t,x,u,\langle\sigma_1(t,x),D_xu\rangle,\cdots,\langle\sigma_d(t,x),D_xu\rangle) \rangle
    \end{align}
\end{thm}
\begin{proof}
    The result follows from
\cite[Lemma~6.2]{fleming2006controlled}, Dini's theorem, and
\cite[Theorem~4.5]{hu2014comparison}.
\end{proof}

\section[Infinite-horizon quadratic G-BSDEs]{Infinite-horizon quadratic $G$-BSDEs}\label{section 4}

Fix $0\le t< \infty,$ $\xi\in \mathbb{L}_G^1(\Omega_t;\mathbb{R}^m),$ and coefficient functions
$b,h_{ij},\sigma_j:[0,\infty)\mapsto \mathbb{R}^m$ and
$f,g_{ij}:[0,\infty)\times \mathbb{R}^m\times \mathbb{R}\times \mathbb{R}^d\mapsto 
\mathbb{R}.$
Consider the following forward SDE and $G$-BSDE: 
\begin{align} \label{infinite SDE1}
     X_s^{t,\xi}&=\xi+\int_t^s b(u,X_u^{t,\xi})\,du+\int_t^s h_{ij}(u,X_u^{t,\xi})\,\langle B^i,B^j\rangle_u+\int_t^s\sigma_j(u,X_u^{t,\xi})\,dB^j_u\,,\;s\ge t
 \\
\label{infinite QBSDE1} Y_s^{t,\xi}&=Y_T^{t,\xi}+\int_s^Tf(u,X_u^{t,\xi},Y_u^{t,\xi},Z_u^{t,\xi})\,du+\int_t^Tg_{ij}(u,X_u^{t,\xi},Y_u^{t,\xi},Z_u^{t,\xi})\,d\langle B^i,B^j\rangle_u\\&-\int_t^TZ_u^{t,\xi}\,dB_u-(K_T^{t,\xi}-K_t^{t,\xi}) \,,\;t\le s\le T<\infty\,.
\end{align} 
When no confusion can arise, we write $X^{t,\xi},Y^{t,\xi},Z^{t,\xi},K^{t,\xi}$ as $X,Y,Z,K,$ respectively, omitting the superscripts $t,\xi.$

\begin{assume}\label{infinite assumption 1} Assume that $b,h_{ij},\sigma_j$ satisfy the following conditions.
\begin{enumerate}
    \item For $1\leq i,j \leq d$, $h_{ij}=h_{ji}$.
    \item The functions $b,h_{ij},\sigma_j$ are continuous in $s$.
    \item There exists a constant $C_1$ such that 
    \begin{align}
        &|b(s,x)-b(s,x')|+\sum_{i,j=1}^d|h_{ij}(s,x)-h_{ij}(s,x')|+\sum_{j=1}^d|\sigma_{j}(s,x)-\sigma_{j}(s,x')|\leq C_1|x-x'|\,,
    \end{align}
    for $s\in[t,\infty),x,x'\in \mathbb{R}^m$.
    \end{enumerate}
\end{assume}

\begin{assume}\label{infinite assumption 2} Assume that $f,g_{ij}$ satisfy the following conditions.
\begin{enumerate}
    \item For $1\leq i,j \leq d$, $g_{ij}=g_{ji}$.
    \item The functions $ f,g$ are continuous in $s$.
    \item There exists a constant $C_1$ such that 
    \begin{align}
    &|f(s,x,y,z)-f(s,x',y',z')|+\sum_{i,j=1}^d|g_{ij}(s,x,y,z)-g_{ij}(s,x',y',z')|\\ &\hspace{1cm}\leq C_1(|x-x'|+|y-y'|+(1+|z|+|z'|)|z-z'|)
    \end{align}
    for $s\in[t,\infty),x,x'\in \mathbb{R}^m$, $y,y'\in \mathbb{R}$, $z,z'\in \mathbb{R}^d$,
    \item There exists a constant $C_2$ such that
    \begin{align}
    |f(s,x,0,0)|+\sum_{i,j=1}^d|g_{ij}(s,x,0,0)|\leq C_2 
    \end{align}
    for $s\in[t,\infty)$, $x\in \mathbb{R}^m$.
    \item There exists a constant $\mu>0$ such that
    \begin{align}
        (f(s,x,y,z)-f(s,x,y',z))(y-y')+2G((g(s,x,y,z)-g(s,x,y',z))(y-y'))\leq -\mu|y-y'|^2
    \end{align} 
    for $s\in[t,\infty),x\in \mathbb{R}^m$, $y,y'\in \mathbb{R}$, $z\in \mathbb{R}^d$.
    \end{enumerate}
\end{assume}

\begin{defi}
A triple $(Y,Z,K)$ is called a solution to \eqref{infinite QBSDE1} on
$[t,\infty)$ if
$(Y,Z)\in\mathbb{S}^2(t,\infty)\times
\mathbb{H}^2(t,\infty;\mathbb{R}^d)$,
$K$ is a decreasing $G$-martingale satisfying $K_t=0$ and
$K_T\in\mathbb{L}_G^2(\Omega_T)$ for every $T>t$, and
\eqref{infinite QBSDE1} holds quasi-surely for all
$t\leq s\leq T<\infty$.
\end{defi}
\begin{thm}\label{infinite horizon Quadartic BSDE exist and unique}
    Suppose Assumptions~\ref{infinite assumption 1} and \ref{infinite assumption 2} hold. Then \eqref{infinite QBSDE1} admits a unique solution $(Y^{t,x},Z^{t,x},K^{t,x})$ in the class for which $Y^{t,x}$ is bounded.
\end{thm}
\begin{proof}
     For simplicity, assume $d=1,m=1,t=0$ and $f=0$.
We first prove uniqueness. Suppose $(Y^1,Z^1,K^1)$ and $(Y^2,Z^2,K^2)$
are solutions to \eqref{infinite QBSDE1}. Since the processes $Y^1,Y^2$
are bounded by a common constant $C$, Theorem~\ref{Z is BMO and K is intgrable}
implies that, for every $0\leq T<\infty$,
     \begin{align}
        \sup_{P\in \mathcal{P}}\sup_{\tau\in 
		\mathcal{T}_{[{s},T]}}\Big|\!\Big| \mathbb{E}^P_{\tau}\Big[\int_{\tau}^T|Z_u^1|^2\,d\langle B \rangle_u\Big]\Big|\!\Big|_{L^\infty(P)} <\infty \,, \quad \sup_{P\in \mathcal{P}}\sup_{\tau\in 
		\mathcal{T}_{[{s},T]}}\Big|\!\Big| \mathbb{E}^P_{\tau}\Big[\int_{\tau}^T|Z_u^2|^2\,d\langle B \rangle_u\Big]\Big|\!\Big|_{L^\infty(P)} <\infty\,.
    \end{align}
 Let $(\hat{Y},\hat{Z},\hat{K}):=(Y^1-Y^2,Z^1-Z^2,K^1-K^2)$ and, for $\epsilon>0$,
define $l^\epsilon(x):=\mathds{1}_{|x|\geq \epsilon}+\frac{|x|}{\epsilon}\mathds{1}_{|x|< \epsilon}$.
Define the following three processes:
\begin{align}
    a_s^{\epsilon}&=l^\epsilon(\hat{Y_s})\frac{g(s,X_s,Y^1_s,Z^1_s)-g(s,X_s,Y_s^2,Z^1_s)}{\hat{Y}_s} -\frac{\mu}{1+\underline{\sigma}^2}(1-l^\epsilon(\hat{Y_s}))\,,\\
    \label{in the proof of Infinite QBSDE thm bmo}b_s^{\epsilon}&=l^\epsilon(\hat{Z}_s)\frac{g(s,X_s,Y^2_s,Z^1_s)-g(s,X_s,Y_s^2,Z_s^2)}{|\hat{Z}_s|^2}\hat{Z}_s+C_1(1-l^\epsilon(\hat{Z}))\,,\\
    m_s^{\epsilon}&=g(s,X_s,Y^1_s,Z^1_s)-g(s,X_s,Y_s^2,Z_s^2)-a^{\epsilon}_s\hat{Y}_s-b_s^{\epsilon}\hat{Z_s}\,.
\end{align}
A direct calculation gives $2G(a_s^{\epsilon})\leq -\mu$,
$|b_s^{\epsilon}|\leq C_1(1+|Z_s^1|+|Z_s^2|)$, and
$|m_s^{\epsilon}|\leq 4\epsilon C_1(1+\epsilon+|Z_s^1|)$.
By Lemma~\ref{lem:infinite-linearization-bmo} below, $b^{\epsilon}$ is a
$G$-BMO martingale generator on every finite time interval.
Let $X^{\epsilon}$ solve the SDE
$X_\cdot^{\epsilon}=1+\int_0^\cdot a_u^{\epsilon}X_u^{\epsilon}d\langle B \rangle _u$.
An argument similar to that used in the proof of
Theorem~\ref{in QBSDE comparison theorem} yields
\begin{align}
    X_s^{\epsilon}(\hat{Y_s}+K^2_s)&=\hat{\mathbb{E}}_s^{b^{\epsilon}}[X_T^{\epsilon}(\hat{Y}_T+K^2_T)+\int_s^TX_u^{\epsilon}( m_u^{\epsilon}-a_u^{\epsilon}K_u^2)\,d\langle B^{\epsilon}\rangle _u]\\
    &\leq \hat{\mathbb{E}}^{b^{\epsilon}}_s[X_T^{\epsilon}\hat{Y}_T+\int_s^TX_u^{\epsilon}m_u^{\epsilon}\,d\langle B^{\epsilon}\rangle_u]+ \hat{\mathbb{E}}_s^{b^{\epsilon}}[X_T^{\epsilon}K^2_T-\int_s^TX_u^{\epsilon}a_u^{\epsilon}K_u^2\,d\langle B^{\epsilon}\rangle _u]\,.
\end{align}
Using $X_s^{\epsilon}K^2_s=\hat{\mathbb{E}}_s^{b^{\epsilon}}[X_T^{\epsilon}K^2_T-\int_s^TX_u^{\epsilon}a_u^{\epsilon}K_u^2\,d\langle B^{\epsilon}\rangle _u]$, we obtain
\begin{align}
    \hat{Y}_s\leq& 2C\hat{\mathbb{E}}^{b^{\epsilon}}_s[e^{\int_s^Ta_u^{\epsilon}\,d\langle B^\epsilon \rangle_u-\int_s^T2G(a_u^{\epsilon})\,du +\int_s^T2G(a_u^{\epsilon})\,du}]\\
    &+\hat{\mathbb{E}}_s^{b^{\epsilon}}[\int_s^T e^{\int_s^ua_v^{\epsilon}\,d\langle B^\epsilon \rangle_v-\int_s^u2G(a_v^{\epsilon})\,dv+\int_s^u2G(a_v^{\epsilon})\,dv}m_u^{\epsilon}\,d\langle B^\epsilon \rangle_u]\\
    \leq& 2C\hat{\mathbb{E}}^{b^{\epsilon}}_s[e^{\int_s^T2G(a_u^{\epsilon})\,du}]+\hat{\mathbb{E}}_s^{b^{\epsilon}}[\int_s^T e^{\int_s^u2G(a_v^{\epsilon})\,dv}m_u^{\epsilon}\,d\langle B^\epsilon \rangle_u]\,.
\end{align}
Letting $\epsilon \rightarrow 0$, we obtain 
\begin{align}
    \hat{Y}_s\leq 2Ce^{-\mu(T-s)}\,.
\end{align}
Letting $T\rightarrow \infty$ gives $Y^1_s-Y^2_s\leq 0$.
Interchanging the two solutions gives $Y^1_s-Y^2_s\geq 0$.
Theorem~\ref{thm:uniqueness} then yields uniqueness of the full solution triple.

We now prove existence. For fixed $n\geq0$, let the triple
$(Y^n,Z^n,K^n)$ be the solution to the following finite-horizon $G$-BSDE:
\begin{align}\label{eqn:fineBS}
    Y_s^n=\int_s^ng(u,X_u,Y_u^n,Z_u^n)\,d\langle B \rangle _u-\int_s^nZ^n_u\,dB_u-(K^n_n-K^n_s)\,.
\end{align}
Here, $0\leq {s}\leq n$.
As in the uniqueness argument, there exist three processes
$a_s^{n,\epsilon}, b_s^{n,\epsilon}$ and $m_s^{n,\epsilon}$ such that 
\begin{align}
    Y_s^n=\int_s^na_u^{n,\epsilon}Y^n_u+b_u^{n,\epsilon} Z_u^n+m_u^{n,\epsilon}+g(u,X_u,0,0)\,d\langle B\rangle _u-\int_s^nZ^n_u\,dB_u-(K^n_T-K^n_s)\,.
\end{align}
Let $\hat{\mathbb{E}}^{b^{n,\epsilon}}$ be the sublinear expectation induced by $b^{n,\epsilon}$.
By \eqref{new_GBM}, the process $B^{n,\epsilon}:=B-\int_0^{\cdot}  b_u^{n,\epsilon}\, d\langle B\rangle _u$
is a $G$-Brownian motion under $\hat{\mathbb{E}}^{b^{n,\epsilon}}$.
It follows that
\begin{align}
    |Y_s^n|&\leq 
    \hat{\mathbb{E}}_s^{ b^{n,\epsilon}}[\int_s^n e^{\int_s^ua_v^{n,\epsilon}\,\langle B^{n,\epsilon} \rangle_v-\int_s^u2G(a_v^{n,\epsilon})\,dv +\int_s^u2G(a_v^{n,\epsilon})\,dv}(m_u^{n,\epsilon}+g(u,X_u,0,0))\,d\langle B^{n,\epsilon} \rangle_u]\\
    &\leq \frac{C_2}{\mu}\,.
\end{align}
We extend the triple $(Y^n,Z^n,K^n)$ to $[0,\infty)$ by setting 
\begin{align}
    Y_s^n=Z_s^n=0, \hspace{0.1cm} K_s^n=K_n^n,\hspace{0.3cm} \forall s>n\,.
\end{align}
For $n\leq \ell$, set $(\hat{Y},\hat{Z},\hat{K})=(Y^n-Y^\ell,Z^n-Z^\ell,K^n-K^\ell)$.
As in the uniqueness argument, there exist three processes
$a_s^{n,\ell,\epsilon}, b_s^{n,\ell,\epsilon}$ and $m_s^{n,\ell,\epsilon}$ such that
\begin{align}
\hat{Y_s}=&\int_s^\ell g(u,X_u,Y_u^n,Z_u^n)-g(u,X_u,Y_u^\ell,Z_u^\ell)-\mathds{1}_{u>n}g(u,X_u,0,0)\,d\langle B\rangle _u \\
&-\int_s^\ell \hat{Z}_u\,dB_u-(\hat{K}_\ell-\hat{K}_s)\\
=&\int_s^\ell a_u^{n,\ell,\epsilon}\hat{Y}_u+b_u^{n,\ell,\epsilon}\hat{Z_u}+m_u^{n,\ell,\epsilon}-\mathds{1}_{u>n}g(u,X_u,0,0)\,d\langle B\rangle_u-\int_s^\ell\hat{Z}_u\,dB_u-(\hat{K}_\ell-\hat{K}_s)\,.
\end{align}
Let $\hat{\mathbb{E}}^{b^{n,\ell,\epsilon}}$ be the sublinear expectation induced by $b^{n,\ell,\epsilon}$.
By \eqref{new_GBM}, the process $B^{n,\ell,\epsilon}:=B-\int_0^{\cdot}  b_u^{n,\ell,\epsilon}\, d\langle B\rangle _u$
is a $G$-Brownian motion under $\hat{\mathbb{E}}^{b^{n,\ell,\epsilon}}$.
It follows that
\begin{align}
    \hat{Y}_s\leq& \hat{\mathbb{E}}_s^{b^{n,\ell,\epsilon}}[\int_s^\ell e^{\int_s^ua_v^{n,\ell,\epsilon}\,d\langle B^{n,\ell,\epsilon} \rangle_v}(m_u^{n,\ell,\epsilon}+\mathds{1}_{u>n}|g(u,X_u,0,0)|)\,d\langle B^{n,\ell,\epsilon} \rangle_u]\\
    \leq&\frac{C_2}{\mu}e^{\mu s}(e^{-\mu n}-e^{-\mu \ell})\,.
\end{align}
Interchanging the two solutions in the same argument yields
\begin{align}
    -\hat{Y_s}\leq&\frac{C_2}{\mu}e^{\mu s}(e^{-\mu n}-e^{-\mu \ell})\,.
\end{align} Thus
\begin{equation}
    \label{eqn:hat_Y}
|\hat{Y}_s|\leq\frac{C_2}{\mu}e^{\mu s}(e^{-\mu n}-e^{-\mu \ell})\,.
\end{equation}
Therefore, for $0\leq T \leq n\leq \ell$, we obtain 
\begin{align}
    \lim_{n,\ell\rightarrow \infty}\hat{\mathbb{E}}[\sup_{0\leq s\leq T}|Y_s^n-Y_s^\ell|^2 ]=0\,.
\end{align}
For fixed $T$, the restriction of $(Y^n,Z^n,K^n)$ to this interval solves the following $G$-BSDE:
\begin{align}
    Y_s^n=Y_T^n+\int_s^Tg(u,Y^n_u,Z^n_u)d\langle B\rangle_u-\int_s^TZ^n_udB_u-(K^n_T-K^n_s)
\end{align}
Since there exists a function $u^n(s,x)$ such that $Y_T^n=u^n(T,X_T)$,
Corollary~\ref{cor needed in Erogidc} yields the existence of a solution to
\eqref{infinite QBSDE1}.   
\end{proof}
\begin{lemma}\label{lem:infinite-linearization-bmo}
    The process $b^{\epsilon}$ defined by \eqref{in the proof of Infinite QBSDE thm bmo} is a $G$-BMO martingale generator.
\end{lemma}
\begin{proof}
For simplicity, assume $d=1$.
Since $b_s^{\epsilon}\leq C_1(1+|Z_s^1|+|Z_s^2|)$, it remains to show that,
for every $T\geq 0$, $b^{\epsilon}\in \mathbb{H}^2(0,T)$.
For each $n\in \mathbb{N}$, define the process 
    \begin{align}
        b_s^{n,\epsilon}=l^\epsilon(\hat{Z}_s)\frac{g^1(Y^2_s,\frac{|Z_s^1|\wedge n}{|Z_s^1|}Z^1_s)-g^1(Y_s^2,\frac{|Z_s^2|\wedge n}{|Z_s^2|}Z_s^2)}{|\hat{Z}_s|^2}\hat{Z}_s+C_1(1-l^\epsilon(\hat{Z}))\,.
    \end{align}
    By Lemma~3.1 of \citet{hu2018ergodic}, for every $T\geq 0$,
$b^{n,\epsilon}\in \mathbb{H}^2(0,T)$. Moreover,
    \begin{align}
        |b_s^{n,\epsilon}-b_s^{\epsilon}|\leq& l^\epsilon(\hat{Z}_s)\frac{|g^1(Y^2_s,Z^1_s)-g^1(Y^2_s,\frac{|Z_s^1|\wedge n}{|Z_s^1|}Z^1_s)|+|g^1(Y^2_s,Z^2_s)-g^1(Y_s^2,\frac{|Z_s^2|\wedge n}{|Z_s^2|}Z_s^2)| }{|\hat{Z}_s|}\\
        \leq&\frac{C_1}{\epsilon}((|Z_s^1|-n)(1+n+|Z_s^1|)\mathds{1}_{|Z_s^1|>n}+(|Z_s^2|-n)(1+n+|Z_s^2|)\mathds{1}_{|Z_s^2|>n})\\
        \leq&\frac{2C_1}{\epsilon}(|Z_s^1|^2\mathds{1}_{|Z_s^1|>n} +|Z_s^2|^2\mathds{1}_{|Z_s^2|>n})\,.
    \end{align}
Using $Z_s^1,Z_s^2\in \mathbb{H}^2(0,T)$ together with
Proposition~3.8 of \citet{li2011stopping}, we conclude that
$b^{\epsilon}\in \mathbb{H}^2(0,T)$. 
From $|b_s^\epsilon|\le C_1(1+|Z_s^1|+|Z_s^2|)$, it follows that 
\begin{align}
	\sup_{P\in \mathcal{P}}\sup_{\tau\in 
		\mathcal{T}_{[{s},T]}}\Big|\!\Big| \mathbb{E}^P_{\tau}\Big[\int_{\tau}^T|b_u^\epsilon|^2\,d\langle B \rangle_u\Big]\Big|\!\Big|_{L^\infty(P)} <\infty\,.
	\end{align}  
    This completes the proof.
\end{proof}

We now relate the finite- and infinite-horizon equations by showing that
finite-horizon solutions converge to the infinite-horizon solution as
the terminal horizon tends to infinity.
Consider the finite-horizon $G$-BSDE 
\begin{align}\label{eqn:appfBSDE} 
Y_s^{t,\xi,T}&=\int_s^Tf(u,X_u^{t,\xi},Y_u^{t,\xi}, Z_u^{t,\xi,T})\,du+\int_s^Tg_{ij}(u,X_u,Y_u^{t,\xi,T},Z_u^{t,\xi,T})\,d\langle B^i,B^j \rangle _u\\
&-\int_s^TZ^{t,\xi,T}_u\,dB_u-(K^{t,\xi,T}_T-K^{t,\xi,T}_s) \,,\;t\le s\le T
\end{align}
and the infinite-horizon $G$-BSDE
\begin{align}\label{eqn:appiBSDE}
Y_s^{t,\xi}&=Y_T^{t,\xi}+\int_s^Tf(u,X_u^{t,\xi},Y_u^{t,\xi},Z_u^{t,\xi})\,du+\int_s^Tg_{ij}(u,X_u^{t,\xi},Y_u^{t,\xi},Z_u^{t,\xi})\,d\langle B^i,B^j\rangle_u\\&-\int_s^TZ_u^{t,\xi}\,dB_u-(K_T^{t,\xi}-K_s^{t,\xi})\,,\;t\le s\le T<\infty
\end{align}

\begin{thm}\label{thm:convergence}
    Suppose Assumptions~\ref{infinite assumption 1} and~\ref{infinite assumption 2} hold.
Let $(Y^{t,\xi,T},Z^{t,\xi,T}, K^{t,\xi,T})$ solve \eqref{eqn:appfBSDE},
and let $(Y^{t,\xi},Z^{t,\xi},K^{t,\xi})$ solve \eqref{eqn:appiBSDE}. 
    Then 
    \begin{align}
        \sup_{t\le u\le s}|Y_u^{t,\xi,T}-Y_u^{t,\xi}|\le \frac{C_2}{\mu}e^{-\mu(T-s)} \quad\text{quasi-surely}
    \end{align}
    for $t\le s\le T<\infty$.  
\end{thm}

\begin{proof}
    This estimate follows from \eqref{eqn:hat_Y} when \(T=n\) is a positive integer satisfying \(n>t\). The extension to arbitrary real \(T>t\) follows by repeating the same estimates with that terminal horizon.
\end{proof}

\section{Bond pricing with endogenous short-rate feedback}\label{sec:app_1}

In this section, we study bond pricing with endogenous short-rate feedback under $G$-expectation.
We formulate the financial market on a $G$-expectation space
$(\Omega,\overline{\mathbb E})$, where the canonical process $W$ is a $d$-dimensional
$G$-Brownian motion under $\overline{\mathbb E}$.
The market consists of the money-market account
\[
M_s
=e^{\int_0^s r_u\,du
+
\int_0^s \rho_{ij,u}\,d\langle W^i,W^j\rangle_u},
\]
where the short-rate processes $(r_s)_{s\geq 0}$ and
$(\rho_{ij,s})_{s\geq 0}$ are specified below, together with zero-coupon
bonds $P(s,T)$ for all maturities $T>0$ satisfying $P(T,T)=1$.

To rule out arbitrage opportunities, following
\citet[Theorem~3.1]{holzermann2022term}, we assume that there exists a
process $\lambda$ such that the discounted bond prices
\[
\frac{P(s,T)}{M_s},
\qquad 0\leq s\leq T,
\]
are symmetric $G$-martingales under the sublinear expectation
$\hat{\mathbb E}:=\hat{\mathbb E}^{\lambda}$.
We recall the definition of $\hat{\mathbb E}^{\lambda}$ in
\eqref{eqn:change}.
The process 
\begin{equation}
\label{new_GBM_pricing}
B_t
:=
W_t-\int_0^t \lambda_u\,d\langle W\rangle_u,
\qquad 0\leq t\leq T,
\end{equation}
is a $G$-Brownian motion under $\hat{\mathbb E}$.
We refer to $\hat{\mathbb E}$ as the risk-neutral sublinear expectation.
Since
$\langle B^i,B^j\rangle
=
\langle W^i,W^j\rangle$,
the money-market account can equivalently be written as
\[
M_s
=
e^{
\int_0^s r_u\,du
+
\int_0^s \rho_{ij,u}\,d\langle B^i,B^j\rangle_u
}.
\]
Accordingly, the time-$s$ price of a zero-coupon bond with maturity $T$
is given by
\begin{equation}
\label{eqn:bond}
P(s,T)
=
\hat{\mathbb E}_s
\left[
e^{
-\int_s^T r_u\,du
-\int_s^T \rho_{ij,u}\,d\langle B^i,B^j\rangle_u
}
\right].
\end{equation}
For notational convenience, we write the logarithmic bond price as
$Y_s^T:=\log P(s,T)$. When the maturity $T$ is clear from the context,
we suppress the superscript $T$ and write $Y_s$.

We now introduce endogenous short-rate feedback.
Specifically, we assume that the short rate is a function of time,
the economic state, and the bond price.
Let $X$ be an economic state process representing macroeconomic factors
such as inflation and unemployment.
We assume that $X$ solves the $G$-SDE \eqref{SDE} with
$t=0$ and $\xi=x\in\mathbb R^m$.
The short-rate processes are then specified by
\[
r_s=r(s,X_s,Y_s),
\qquad
\rho_{ij,s}=\rho_{ij}(s,X_s,Y_s),
\qquad s\geq 0.
\]
Substituting these feedback specifications into \eqref{eqn:bond} yields
\begin{align}
\label{eqn:appli_1}
e^{Y_s}
=
\hat{\mathbb E}_s
\left[
e^{
-\int_s^T r(u,X_u,Y_u)\,du
-\int_s^T \rho_{ij}(u,X_u,Y_u)\,
d\langle B^i,B^j\rangle_u
}
\right].
\end{align}
Unlike the standard bond-pricing formula, \eqref{eqn:appli_1} is implicit:
the short rate depends on the bond price through $Y$.
The existence and uniqueness of a process $Y$ satisfying
\eqref{eqn:appli_1} therefore require a separate argument.

\begin{assume}\label{assume:r}
Suppose that there exist positive constants $C_1$ and $C_2$ such that
\begin{align}
&|r(s,x,y)-r(s,x',y')|
+
\sum_{i,j=1}^d
|\rho_{ij}(s,x,y)-\rho_{ij}(s,x',y')|
\leq
C_1\bigl(|x-x'|+|y-y'|\bigr),
\\
&|r(s,x,0)|
+
\sum_{i,j=1}^d
|\rho_{ij}(s,x,0)|
\leq C_2
\end{align}
for all $s\in[0,T]$, $x,x'\in\mathbb R^m$, and $y,y'\in\mathbb R$.
\end{assume}

We seek to establish the existence and uniqueness of a process $Y$
satisfying \eqref{eqn:appli_1}. To this end, we transform the implicit
bond-pricing equation \eqref{eqn:appli_1} into the following quadratic $G$-BSDE:
\begin{equation}
\label{eqn:bsde}
\begin{aligned}
Y_s
={}&
-\int_s^T r(u,X_u,Y_u)\,du
-\int_s^T
\left(
\rho_{ij}(u,X_u,Y_u)
-\frac{1}{2}Z_u^iZ_u^j
\right)
d\langle B^i,B^j\rangle_u
\\
&\quad
-\int_s^T Z_u\,dB_u
-(K_T-K_s),
\qquad 0\leq s\leq T.
\end{aligned}
\end{equation}
The following theorem establishes the well-posedness of this $G$-BSDE
and, consequently, of the endogenous bond-pricing equation
\eqref{eqn:appli_1}.

\begin{thm}
Suppose Assumptions~\ref{assumption} and~\ref{assume:r} hold.
Then the following statements hold.
\begin{enumerate}
\item[(i)]
The $G$-BSDE \eqref{eqn:bsde} admits a unique solution
$(Y,Z,K)
\in
\mathbb S^{2}(0,T)
\times
\mathbb H^{2}(0,T;\mathbb R^d)
\times
\mathbb L_G^2(\Omega_T)$
such that $Y$ is bounded.
\item[(ii)]
The process $Y$ is the unique bounded process in
$\mathbb S^2(0,T)$ satisfying \eqref{eqn:appli_1}.
\end{enumerate}
\end{thm}
 
\begin{proof} 
    For simplicity, assume $d=1$. The coefficients $r$ and
$\rho$ induce generators satisfying the assumptions of
Theorem~\ref{thm:exist}, which yields a unique solution $(Y,Z,K)$
to \eqref{eqn:bsde}.

We first verify that $Y$ satisfies \eqref{eqn:appli_1}. Applying
It\^{o}'s formula to
$e^{Y_s-\int_0^s r(u,X_u,Y_u)\,du-\int_0^s
\rho(u,X_u,Y_u)\,d\langle B\rangle_u}$, we obtain
    \begin{align}
        e^{Y_s-\int_0^sr(u,X_u,Y_u)\,du-\int_0^s\rho(u,X_u,Y_u)\,d\langle B\rangle_u}=\hat{\mathbb{E}}_s[e^{-\int_0^Tr(u,X_u,Y_u)\,du-\int_0^T\rho(u,X_u,Y_u)\,d\langle B\rangle_u}]\,.
    \end{align}
    Therefore, $Y$ satisfies \eqref{eqn:appli_1}.
    To prove uniqueness for the implicit pricing equation, let
$Y'\in\mathbb{S}^2(0,T)$ be another bounded process satisfying
\eqref{eqn:appli_1}. Then, for each
$0\leq s\leq T$,
    \begin{align}
         e^{Y'_s-\int_0^sr(u,X_u,Y'_u)\,du-\int_0^s\rho(u,X_u,Y'_u)\,d\langle B\rangle_u}=\hat{\mathbb{E}}_s[e^{-\int_0^Tr(u,X_u,Y'_u)\,du-\int_0^T\rho(u,X_u,Y'_u)\,d\langle B\rangle_u}]\,.
    \end{align}
    By the $G$-martingale representation theorem, there exist
    $Z\in\mathbb{H}^2(0,T)$ and a decreasing $G$-martingale $K$ such that
\begin{align}
        e^{Y'_s-\int_0^sr(u,X_u,Y'_u)\,du-\int_0^s\rho(u,X_u,Y'_u)\,d\langle B\rangle_u}=\int_0^sZ_u\,dB_u+K_s\,.
    \end{align}
    Applying It\^{o}'s formula again yields
    \begin{align}
        Y'_s&=
        -\int_s^Tr(u,X_u,Y'_u)\,du-\int_s^T\rho(u,X_u,Y'_u)\,d\langle B\rangle_u\\
        &-\int_s^T\frac{1}{2}(e^{-Y'_u+\int_0^ur(\ell,X_\ell,Y'_\ell)\,d\ell+\int_0^u\rho(\ell,X_\ell,Y'_\ell)\,d\langle B\rangle_\ell}Z_u)^2\,d\langle B\rangle_u\\
        &-\int_s^T e^{-Y'_u+\int_0^ur(\ell,X_\ell,Y'_\ell)\,d\ell+\int_0^u\rho(\ell,X_\ell,Y'_\ell)\,d\langle B\rangle_\ell}Z_u\,dB_u
        -\int_s^T e^{-Y'_u+\int_0^ur(\ell,X_\ell,Y'_\ell)\,d\ell+\int_0^u\rho(\ell,X_\ell,Y'_\ell)\,d\langle B\rangle_\ell}dK_u\,.
    \end{align}
    Define
    \begin{align}
        Z'_s&=e^{-Y'_s+\int_0^sr(u,X_u,Y'_u)\,du+\int_0^s\rho(u,X_u,Y'_u)\,d\langle B\rangle_u}Z_s\,,\\
        K'_s&=\int_0^s e^{-Y'_u+\int_0^ur(\ell,X_\ell,Y'_\ell)\,d\ell+\int_0^u\rho(\ell,X_\ell,Y'_\ell)\,d\langle B\rangle_\ell}dK_u\,.
    \end{align}
    The process
\[
e^{-Y'_u+\int_0^u r(\ell,X_\ell,Y'_\ell)\,d\ell
+\int_0^u\rho(\ell,X_\ell,Y'_\ell)\,d\langle B\rangle_\ell}
\]
is positive and bounded. Hence, by
Lemma~3.4 of \citet{hu2014backward}, $K'$ is a decreasing
$G$-martingale, and $(Y',Z',K')$ is a solution to
\eqref{eqn:bsde}. Uniqueness in Theorem~\ref{thm:exist} yields
$(Y,Z,K)=(Y',Z',K')$. Consequently, $Y$ is the unique bounded
process in $\mathbb{S}^2(0,T)$ satisfying \eqref{eqn:appli_1}.
\end{proof}

\section{Policy short-rate design for target bond prices}

\subsection{Fixed-maturity bonds}

We consider the inverse problem of choosing a policy short rate to reproduce a prescribed bond-price target.
Let $\psi:[0,T]\times\mathbb{R}^m\to\mathbb{R}$ be a prescribed function, and suppose that the monetary authority seeks to attain the target bond price
\[
P(s,T)=e^{\psi(s,X_s)},
\qquad 0\leq s\leq T,
\]
for a zero-coupon bond with maturity $T$.
This formulation allows the target to depend on macroeconomic conditions, such as inflation and unemployment, through the state process $X$.

For a fixed maturity, the construction follows by matching the dynamics
of the target price. We present this result as preparation for the
long-maturity problem in Subsection~\ref{sec:long-maturity-bonds}.
The proposition below follows directly from It\^{o}'s formula,
so we omit its proof.

\begin{prop}
Let $\psi\in C^{1,2}([0,T]\times\mathbb{R}^m)$ satisfy
$\psi(T,\cdot)=0$.
Define
\[
r(s,x)
:=
\partial_s\psi(s,x)
+
\left\langle
D_x\psi(s,x),
b(s,x)
\right\rangle
\]
and
\begin{align}
\rho_{ij}(s,x)
&:=
\left\langle
D_x\psi(s,x),
h_{ij}(s,x)
\right\rangle
+
\frac{1}{2}
\left\langle
D_x^2\psi(s,x)\sigma_i(s,x),
\sigma_j(s,x)
\right\rangle
\nonumber\\
&\quad
+
\frac{1}{2}
\left\langle
D_x\psi(s,x),
\sigma_i(s,x)
\right\rangle
\left\langle
D_x\psi(s,x),
\sigma_j(s,x)
\right\rangle
\end{align}
for $(s,x)\in[0,T]\times\mathbb{R}^m$.
Suppose Assumption~\ref{assumption} holds and that
$r$ and $\rho_{ij}$ are bounded; that is, there exists a constant
$C>0$ such that
\[
|r(s,x)|
+
\sum_{i,j=1}^d
|\rho_{ij}(s,x)|
\leq C
\]
for all $(s,x)\in[0,T]\times\mathbb{R}^m$.
Then the target satisfies
\begin{align}
e^{\psi(s,X_s)}
=
\hat{\mathbb{E}}_s
\left[
e^{
-\int_s^T r(u,X_u)\,du
-\int_s^T
\rho_{ij}(u,X_u)\,
d\langle B^i,B^j\rangle_u
}
\right],
\qquad 0\leq s\leq T.
\end{align}
\end{prop}

\subsection{Long-maturity bonds}
\label{sec:long-maturity-bonds}

The preceding subsection considers an inverse problem for a single,
fixed maturity. We now study the corresponding long-maturity problem.
Suppose that the monetary authority specifies a target long-term
yield $\lambda>0$. To remove the linear decay in the logarithmic
price associated with this target, we define the compensated
logarithmic bond price
\begin{equation}
\label{eqn:compensated-log-bond-price}
Y_s^T
:=
\log P(s,T)+\lambda(T-s),
\qquad 0\leq s\leq T.
\end{equation}
Thus,
\[
e^{Y_s^T}
=
e^{\lambda(T-s)}P(s,T)
\]
represents the bond price after removing the prescribed exponential
decay with respect to maturity.

The feedback coefficients now depend on $Y^T$, rather than on the
uncompensated logarithmic bond price. The resulting bond-price
process is therefore required to satisfy
\begin{equation}
\label{eqn:long-maturity-bond-pricing}
P(s,T)
=
\hat{\mathbb E}_s
\left[
e^{
-\int_s^T r(u,X_u,Y_u^T)\,du
-\int_s^T
\rho_{ij}(u,X_u,Y_u^T)\,
d\langle B^i,B^j\rangle_u
}
\right].
\end{equation}
By the positive homogeneity of $\hat{\mathbb E}_s$,
this equality is equivalent to
\begin{equation}
\label{eqn:compensated-bond-pricing}
e^{Y_s^T}
=
\hat{\mathbb E}_s
\left[
e^{
-\int_s^T
(r(u,X_u,Y_u^T)-\lambda)\,du
-\int_s^T
\rho_{ij}(u,X_u,Y_u^T)\,
d\langle B^i,B^j\rangle_u
}
\right].
\end{equation}

Let $\phi:\mathbb R^m\to\mathbb R$ be the target profile for the
compensated logarithmic bond price. The authority seeks feedback
coefficients such that, for each fixed $s\geq0$,
\begin{equation}
\label{eqn:long-maturity-targets}
\lim_{T\to\infty}
e^{\lambda(T-s)}P(s,T)
=
e^{\phi(X_s)}
\quad\text{and}\quad
-\lim_{T\to\infty}
\frac{\log P(s,T)}{T-s}
=
\lambda,
\qquad\text{quasi-surely}.
\end{equation}
The first limit specifies the state-dependent profile that remains
after removing the prescribed long-maturity decay. It also implies
the long-term yield identity in the second limit.

\begin{thm} Fix $\lambda>0$, a function $\phi\in C^2(\mathbb{R}^m)$,
and a constant $\mu>0$. Define
    $$r(s,x,y):=\left\langle D\phi(x),b(s,x)\right\rangle+\mu(y-\phi(x))+\lambda$$
    and 
    $$\rho_{ij}(s,x,y)=\left\langle D\phi(x),h_{ij}(x)\right\rangle+\frac{1}{2}
\left\langle
D^2\phi(x)\sigma_i(s,x),
\sigma_j(s,x)
\right\rangle+\frac{1}{2}\left\langle D\phi(x),\sigma_i(s,x)\right\rangle\left\langle D\phi(x),\sigma_j(s,x)\right\rangle\,.$$ 
Suppose that Assumption~\ref{infinite assumption 1} holds and that there exist positive constants $C_1, C_2$ such that
\begin{align}
    |r(s,x,0)-r(s,x',0)|+|\rho_{ij}(s,x)-\rho_{ij}(s,x')|\le C_1|x-x'|\,,\quad|r(s,x,0)|+|\rho_{ij}(s,x,0)|\le C_2
\end{align}
for all $s\in[0,\infty)$ and $ x,x'\in \mathbb{R}^m$.
Then the following statements hold.
\begin{enumerate}[label=(\roman*)]
    \item The finite-horizon quadratic $G$-BSDE 
\begin{align}
    Y_s^T&=-\int_s^T r(u,X_u,Y_u^T)-\lambda\,du-\int_s^T \rho_{ij}(u,X_u,Y^T_u)-\frac{1}{2}(Z_u^T)^i(Z_u^T)^j\,d\langle B^i, B^j\rangle_u\\
    &-\int_s^TZ^T_u\,dB_u-(K^T_T-K_s)
\end{align}
admits a unique solution $(Y^T,Z^T,K^T)$ with $Y^T$ bounded.
\item The process $Y^T$ is the unique bounded process in $\mathbb{S}^2(0,T)$ satisfying \eqref{eqn:compensated-bond-pricing}.
\item The infinite-horizon quadratic $G$-BSDE 
\begin{align}
    Y_s&=Y_T-\int_s^Tr(u,X_u,Y_u)-\lambda \,du-\int_s^T\rho_{ij}(u,X_u,Y_u)-\frac{1}{2}Z_u^iZ_u^j\,\langle B^i,B^j\rangle_u\\
    &\quad-\int_s^TZ_u\,dB_u-(K_T-K_t)
\end{align}
admits a unique solution $(Y,Z,K)$ with $Y$ bounded. Moreover, this solution satisfies $Y=\phi(X)$.
\item The processes $Y^T$ and $Y$ satisfy the estimate 
\begin{equation}
    \label{eqn:long-maturity-convergence}
\sup_{0\le u\le s}|Y_u^T-Y_u| \le \frac{C_2}{\mu}e^{-\mu(T-s)} \quad\text{quasi-surely}\,.
\end{equation}
Consequently, both limits in
\eqref{eqn:long-maturity-targets} hold.
\end{enumerate}

\end{thm}

The parameter $\mu$ in the feedback term
$\mu(y-\phi(x))$ controls the rate at which the finite-maturity solution
approaches the prescribed long-maturity profile. When the compensated
logarithmic bond price exceeds its target, the term increases the
specified policy rate; when the price falls below its target, the term
decreases the rate. As shown by \eqref{eqn:long-maturity-convergence},
convergence is exponential with rate $\mu$ as maturity increases.
Thus, selecting $\mu$ determines the rate in the horizon-convergence
estimate. This estimate concerns convergence with respect to maturity,
not calendar-time adjustment after a policy intervention.

\section{Conclusion}
\label{sec:con}

This paper studies bond valuation when the discounting rule depends on the bond price itself and the underlying state evolves under volatility uncertainty. The feedback leads to an implicit pricing problem, rather than the evaluation of a payoff under an independently specified short rate. Representing the logarithmic bond price as the solution of a quadratic $G$-BSDE allows us to analyze this self-consistency problem through existence, uniqueness, comparison, and stability results.

The finite-horizon analysis yields a unique bounded valuation under suitable assumptions, together with BMO estimates and Markovian regularity. Under an additional strict monotonicity condition, the infinite-horizon equation admits a unique bounded solution, and finite-horizon approximations converge exponentially on compact time intervals. These results address both the well-posedness of the feedback problem at a fixed maturity and its behavior as maturity tends to infinity.

The inverse constructions determine discount-rate coefficients from a prescribed price profile. At a fixed maturity, admissible smooth targets are reproduced by matching their $G$-It\^{o} dynamics. For long maturities, a prescribed yield and an admissible bounded state-dependent profile are combined with linear restoring feedback so that the compensated logarithmic price approaches the target exponentially. The feedback coefficient determines the rate of convergence with respect to maturity, while the target yield determines the asymptotic exponential decay of the bond price. These results concern implementability within the specified model, not welfare-optimal policy.

Several questions remain outside the scope of these valuation results. The sublinear pricing expectation is specified as part of the model. Uniqueness of its pricing fixed point does not by itself establish a unique equilibrium price or the symmetric-martingale conditions needed to rule out arbitrage in a traded bond market. Extending feedback from a reference bond to several maturities under a common short-rate process is therefore a natural next step. Other directions include restrictions on admissible rate coefficients, numerical approximation of the pricing fixed point, and statistical specification of the volatility-uncertainty set. The present analysis provides a basis for these extensions by connecting robust valuation, horizon stability, and target-oriented short-rate design within a common backward-equation framework.

\noindent\textbf{Acknowledgements.}\\ 
Hyungbin Park was supported by the National Research Foundation of
Korea (NRF) grants funded by the Ministry of Science and ICT (2021R1C1C1011675, 2022R1A5A6000840, RS-2026-25488333). Financial support from the Institute for Research in Finance and
Economics of Seoul National University is gratefully acknowledged.


\end{document}